\documentclass[reprint,superscriptaddress,amsmath,amssymb,aps,floatfix]{revtex4-2}

\usepackage{xcolor}
\usepackage{amsthm}
\usepackage{needspace}
\usepackage{physics}
\usepackage{graphicx}
\usepackage{dcolumn}
\usepackage{bm}
\usepackage[colorlinks=true,linkcolor=blue,citecolor=blue,urlcolor=blue]{hyperref}
\usepackage{xcolor}

\newtheorem{definition}{Definition}
\newtheorem{theorem}{Theorem}
\usepackage{weingarten}
\usepackage{placeins}
\begin{document}

\title[Optimal linear-optical interferometers for unitary t-designs]{Optimal linear-optical interferometers for unitary t-designs in a multi-photon Fock space}

\author{R.A. Grinshtein}
\email{grinshtein.ra21@physics.msu.ru}
\affiliation{
Faculty of Physics and Quantum Technologies Centre, M. V. Lomonosov Moscow State University, Leninskie Gory 1, Moscow, Russia, 119991
}%

\author{I.V. Dyakonov}
\affiliation{
Faculty of Physics and Quantum Technologies Centre, M. V. Lomonosov Moscow State University, Leninskie Gory 1, Moscow, Russia, 119991
}%
\affiliation{
Russian Quantum Center, 30 Bolshoy bul’var building 1, Moscow 121205, Russia
}%

\author{S.S. Straupe}
\affiliation{
Faculty of Physics and Quantum Technologies Centre, M. V. Lomonosov Moscow State University, Leninskie Gory 1, Moscow, Russia, 119991
}%
\affiliation{
Russian Quantum Center, 30 Bolshoy bul’var building 1, Moscow 121205, Russia
}%

\date{\today}

\begin{abstract}
We investigate the construction of unitary \(1\)-designs and two-photon designs on programmable linear-optical interferometers. 
Using the frame-potential loss as the optimization objective, we identify architectures that reproduce the relevant Haar moments with considerably fewer programmable phase shifters than required for universal linear optics. 
For unitary \(1\)-designs, we provide an exact construction and prove that the minimum number of programmable phase shifters within the considered layered architecture is \(2(d-1)\), where \(d\) is the number of optical modes. 
Numerical optimization confirms this parameter count over the tested range of dimensions. For two-photon designs, defined by first-moment matching in the symmetric two-photon representation of \(U(d)\), we numerically investigate the resources required to achieve prescribed approximation thresholds. 
We benchmark the resulting ensembles using single- and two-photon transition probabilities and average channel fidelities. Furthermore, we develop a protocol for estimating the total photon number of a multiphoton Fock state using a single monitored output mode and derive the corresponding Haar-averaged photon-count moments. 
Our findings demonstrate that shallow interferometers can reproduce Haar-averaged observables with reduced hardware resources, 
while photon-number fluctuations provide a more stringent test of approximate two-photon designs.
\end{abstract}

\maketitle

\section{Introduction}
Linear-optical platforms remain one of the leading approaches to quantum information processing due to their low noise, scalability, and compatibility with integrated photonic technologies \cite{kok2007linear, o2007optical, bourassa2021blueprint}.
In particular, programmable interferometers play a central role in photonic quantum computing \cite{madsen2022linopt, bouchard2024, slussarenko2019photonic}.
 Tasks such as boson sampling \cite{gard2015introduction, spring2013boson, broome2013photonic} and the modeling of multiphoton random walks \mbox{
\cite{shou2026entanglement, tang2022generating} }\hskip0pt
motivate the study of optical ensembles with controlled random-unitary statistics.
This is because Haar-random unitaries provide a natural notion of uniform randomness on the unitary group.

In practice, however, exact sampling from the Haar measure is difficult to implement efficiently on near-term devices.
Standard methods for synthesizing arbitrary linear-optical transfer matrices (i.e., operations in the mode space), such as the Reck \cite{PhysRevLett.73.58} or Clements \cite{clements2016optimal} architectures, require a quadratic scaling of tunable elements, namely $\mathcal{O}(d^2)$ phase shifters for $d$ modes.
In experimental practice, configuring and maintaining deep optical networks introduces hardware challenges.
Cumulative propagation losses scale with the circuit depth \cite{bao2023, ehrhardt2022}, which exponentially suppresses the coincidence photon registration rates in multi-photon tasks. Furthermore, dense arrays of thermo-optic phase shifters suffer from thermal crosstalk and phase drifts \cite{cavicchioli2026}, necessitating complex, time-consuming multi-parameter calibration loops.
This hardware overhead severely limits the scalability of exact Haar-random operations on integrated photonic chips.
Recent studies have also investigated the complexity of shallow boson-sampling circuits \cite{go2024exploring}.
The statistical properties of such circuits depend on their connectivity and depth. Here, our aim is to match specified low-order Haar moments, rather than to establish the computational hardness of the resulting sampling problem.
It is therefore important to construct robust, finite-depth optical architectures that reproduce Haar-like statistics up to a prescribed order while using significantly fewer physical resources.

Such ensembles are formalized by the notion of a unitary $t$-design. A unitary $t$-design reproduces the first $t$ moments of the Haar measure. An analogous idea appears in classical Euclidean geometry. Consider a polynomial $p_t(\boldsymbol{x}) = \sum_{|\boldsymbol{\alpha}| \leq t} c_{\boldsymbol{\alpha}} \boldsymbol{x}^{\boldsymbol{\alpha}}$, where $\boldsymbol{x} = (x_1, \dots, x_n) \in \mathbb{R}^{n}$ and $\boldsymbol{\alpha} = (\alpha_1, \dots, \alpha_n) \in \mathbb{N}^{n}$, so $\boldsymbol{x}^{\boldsymbol{\alpha}} = x_1^{\alpha_1}x_2^{\alpha_2} \dots x_n^{\alpha_n}$. Instead of integrating this polynomial over the unit sphere $\mathcal{S}^{n-1}$ with respect to the normalized spherical measure $\mu_\mathcal{S}$, one may average it over a finite set of points $\mathcal{K}$, called a \textit{spherical $t$-design}, and obtain the same result for all polynomials of degree at most $t$. In the simplest three-dimensional examples, the vertices of a tetrahedron form a spherical $2$-design, while the vertices of a cube form a spherical $3$-design \cite{colbourn2006handbook}.

The notion of a unitary $t$-design can be defined rigorously as follows \cite{mele2024introduction, gross2007evenly}:

\begin{definition}{\label{def}}
Let $\mathcal{K} = \{U_i\}$ be an ensemble of unitary matrices $U_i \in U(d)$. We say that $\mathcal{K}$ forms a unitary $t$-design if, for every polynomial $p_{t,t}(U,U^\dagger)$ of degree at most $t$ in the matrix elements of $U$ and at most $t$ in the matrix elements of $U^\dagger$, the average over $\mathcal{K}$ coincides with the average over the Haar measure on $U(d)$:
\begin{equation}{\label{eq:1}}
 \frac{1}{|\mathcal{K}|}\sum_i p_{t,t}(U_i,U_i^\dagger) = \int_{U(d)} p_{t,t}(U,U^\dagger)\, d\mu_{\mathrm{H}}(U).
\end{equation}
Equivalently, an ensemble $\mathcal{K}$ is a unitary $t$-design if its $t$-th moment operator matches that of the Haar measure:
\begin{equation}{\label{eq:2}}
 \frac{1}{|\mathcal{K}|}\sum_i U_i^{\otimes t} \otimes (U_i^\ast)^{\otimes t} = \int_{U(d)} U^{\otimes t} \otimes (U^\ast)^{\otimes t}\, d\mu_{\mathrm{H}}(U).
\end{equation}
\end{definition}

Analogous work is being conducted in the realm of qubit-based quantum computing, where unitary transformations must be synthesized directly within the multi-qubit state space.
In this context, unitary $t$-designs in $n$-qubit systems with $d = 2^n$ have been studied extensively.
 In particular, the $n$-qubit Pauli group $\mathcal{P}_n = \{\mathbb{I}, X, Y,Z\}^{\otimes n}$ forms a unitary $1$-design. The Clifford group is the normalizer of the Pauli group,
\begin{equation}
\mathrm{Cl}(n)=\{U\in U(2^n):U\mathcal P_nU^\dagger=\mathcal P_n\},
\end{equation}
where $\mathcal P_n$ includes the scalar phases required for group closure. The qubit Clifford group forms a unitary $t$-design for $t\leq3$ \mbox{
\cite{webb2015clifford}}\hskip0pt
.
Implementing a logical multiqubit Clifford circuit is a different control problem from programming a mode-space interferometer; in particular, passive linear optics alone does not provide deterministic entangling gates for dual-rail photonic qubits without additional resources \cite{kok2007linear, PhysRevA.68.032310}.

\begin{figure*}[t]
 \centering
 \includegraphics[width=\textwidth]{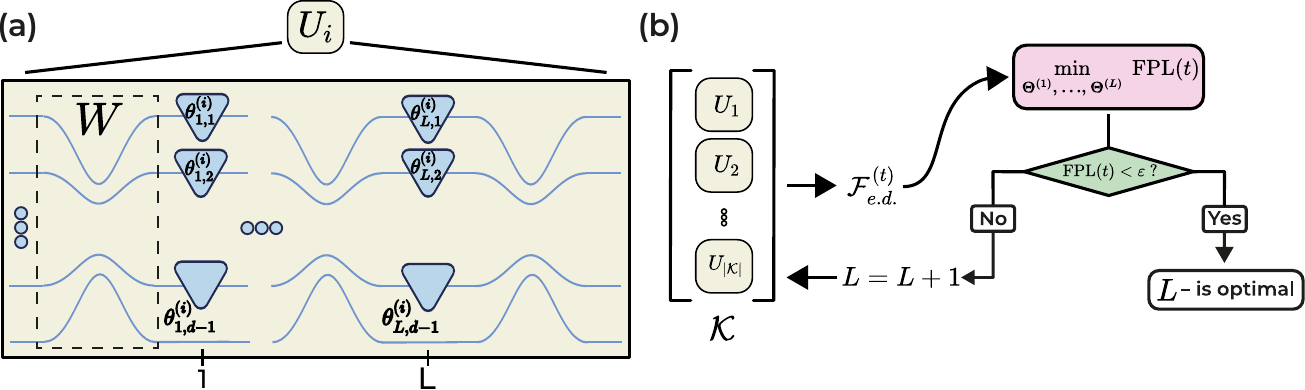}
\caption{(a) Schematic of the programmable linear-optical interferometer used to generate matrices from the ensemble $\mathcal{K}$. Static mode-mixing layers $W$ alternate with programmable phase-shifter layers $\Phi_l(\boldsymbol{\theta}_l)$, and each choice of phases $\boldsymbol{\Theta}^{(i)}$ produces a unitary matrix $U_i$. (b) Optimization workflow for constructing an $\varepsilon$-approximate unitary $t$-design. For a fixed number of programmable layers $L$, the phase settings for each matrix from an ensemble are optimized by minimizing the frame potential loss, and the procedure is repeated until the target design condition is satisfied.}
 \label{fig:1}
\end{figure*}

Based on the previous paragraph, we distinguish the mode space from each fixed-photon-number sector. A passive $d$-mode interferometer is specified by $U\in U(d)$. Its action on $n$ indistinguishable photons is the symmetric-power representation
\begin{equation}
R_n(U)=\left.U^{\otimes n}\right|_{\mathrm{Sym}^n(\mathbb C^d)},
\qquad D_n=\binom{n+d-1}{n}.
\end{equation}
Matching moments of $U$ on $U(d)$ and matching moments of $R_n(U)$ on $U(D_n)$ are distinct requirements. In this work, the one- and two-photon objectives seek a unitary $1$-design of $R_1(U)$ and $R_2(U)$, respectively. We use ``two-photon design'' as an operational label for the second requirement, not as a synonym for a full unitary $2$-design.
For $n,d\ge2$, passive linear optics can form a unitary $1$-design on the fixed-$n$ sector, but cannot form an exact unitary $2$-design on the full group $U(D_n)$ \mbox{
\cite{saied2024advancing}}\hskip0pt
.

For programmable linear-optical interferometers \cite{PhysRevLett.73.58, clements2016optimal, saygin2020robust}, the complexity of implementing a given ensemble in multi-photon Fock space is strongly constrained by hardware resources. In particular, the number of programmable phase shifters determines the effective dimension of the accessible family of interferometers and directly affects fabrication complexity, control overhead, and robustness to noise. For example, constructing an arbitrary unitary $U \in SU(d)$ requires $d^2-1$ real parameters. This naturally raises the question of how many controllable parameters are actually necessary to realize an ensemble that still behaves as a unitary $t$-design for defined system settings: number of modes $d$ and photons $n$.

In this paper, we study the construction of one-/two-photon designs using programmable linear-optical interferometers with a reduced number of tunable phase shifters. We formulate the search as frame-potential optimization for ensembles of interferometers with fixed mode-mixing layers. For one-photon designs, we give an exact construction and prove that two complete programmable phase layers, containing $2(d-1)$ independent phases per interferometer, are necessary and sufficient within the considered architecture. For two-photon designs, we numerically investigate the number of programmable phases needed to reach prescribed frame-potential thresholds. We assess the resulting ensembles through transition probabilities, average channel fidelities, and photon-number estimation.

The paper is organized as follows. In \hyperref[sec:problem]{Sec.~II}, we formulate the problem and introduce the optimization criterion.
In \hyperref[sec:optimal]{Sec.~III}, we present one-/two-photon designs and numerical results.
In \hyperref[sec:benchmark]{Sec.~IV}, we benchmark the proposed constructions on representative photonic tasks.
Finally, in \hyperref[sec:discussion]{Sec.~V}, we summarize the results and discuss future directions.

\section{Problem formulation}\label{sec:problem}
In this work, we focus on the programmable interferometer architecture introduced in \cite{saygin2020robust}, characterized by alternating static mode-mixing layers and programmable phase shifters.
We employ this architecture due to its inherent robustness to fabrication imperfections and its tractable structure for numerical optimization.
 We restrict attention to architectures with identical fixed mixing layers.
Consequently, the unitary matrix $U \in U(d)$ parameterized by an interferometer with $L$ programmable phase-shifter layers is given by
\begin{equation}\label{eq:3}
 U\left(\boldsymbol{\Theta}\right) = W\Phi_{L}(\boldsymbol{\theta}_L)W\dots W\Phi_1(\boldsymbol{\theta}_1)W,
\end{equation}
where $\Phi_l(\boldsymbol{\theta}_l) = \mathrm{diag}\left(e^{i\theta_{l,1}}, \dots, e^{i\theta_{l,d-1}}, 1\right)$ for $l = 1, \dots, L$, $\boldsymbol{\Theta} = \left(\boldsymbol{\theta}_1, \dots, \boldsymbol{\theta}_L\right)$, and the fixed mixing matrix $W$ is chosen to be the discrete Fourier transform (DFT) matrix because it provides balanced mode mixing. Repeated application of the same balanced DFT matrix, interleaved with programmable phase shifts, is mathematically sufficient to span the required unitary space. It is worth noting that the matrix $W$ can be chosen to be any other unitary matrix that exhibits mode-mixing properties.

For each unitary matrix $U_i$ in a unitary $t$-design ensemble $\mathcal{K}$, we denote by $\boldsymbol{\Theta}^{(i)}$ the set of phase-shifter values required to realize $U_i$ with the interferometer. The corresponding architecture is illustrated in Fig.~\ref{fig:1}(a). To quantify how close an ensemble is to an exact unitary $t$-design, we use the frame potential \cite{mele2024introduction}. For the Haar measure, it is given by
\begin{multline}\label{eq:4}
 \mathcal{F}_{\mu_H}^{(t)} = \mathrm{Tr} \left[\int_{U(d)} U^{\otimes t} \otimes (U^\ast)^{\otimes t}\, d\mu_{\mathrm{H}}(U)\right] = \\
 \int_{U(d)} \left|\mathrm{Tr}(U)\right|^{2t}\, d\mu_{\mathrm{H}}(U),
\end{multline}
For an arbitrary equally weighted finite unitary ensemble, the frame potential is
\begin{equation}\label{eq:5}
 \mathcal{F}^{(t)}_{\mathrm{e.d.}} = \frac{1}{|\mathcal{K}|^2} \sum_{i,j = 1}^{|\mathcal{K}|} \left|\mathrm{Tr}\left(U_i U_j^{\dagger}\right)\right|^{2t}.
\end{equation}

The frame potential satisfies several useful properties. In particular, for $t \leq d$ one has $\mathcal{F}^{(t)}_{\mathrm{e.d.}} \geq \mathcal{F}^{(t)}_{\mu_H} = t!$ \cite{mele2024introduction}. This motivates our \textit{frame potential loss} (FPL), defined as
\begin{equation}\label{eq:6}
 \mathrm{FPL}(t) = \left(\mathcal{F}^{(t)}_{\mathrm{e.d.}} - \mathcal{F}^{(t)}_{\mu_H}\right)^2,
\end{equation}
which we require to be bounded by a prescribed threshold $\varepsilon$, i.e., $\mathrm{FPL}(t) < \varepsilon$. If this condition is satisfied, the ensemble $\mathcal{K}$ is treated as an $\varepsilon$-approximate unitary $t$-design.

Certain conditions can be obtained analytically, providing a foundation upon which the construction of the investigated structures can proceed.
Lower bounds and exact minimal sizes of unitary designs are known for several important cases \cite{gross2007evenly, roy2009unitary}. For example, a unitary $1$-design requires at least $|\mathcal{K}| = d^2$ elements, whereas a unitary $2$-design requires at least $|\mathcal{K}| = d^4 - 2d^2 + 2$ elements.
These results remove the ambiguity in the required ensemble size for the low-order designs considered in this work.
The first is attainable in every dimension; the second is a lower bound and does not by itself guarantee a design of that cardinality.

Our goal is therefore to determine the minimal number of programmable interferometer parameters in Eq.~\eqref{eq:3} for which the generated ensemble $\mathcal{K}$ becomes an $\varepsilon$-approximate unitary $t$-design.

The optimization workflow is shown in Fig.~\ref{fig:1}(b).
Due to the intrinsic geometry of the unitary group and the nonlinear dependence on the optical phases, the optimization landscape of the FPL($t$) is highly non-convex.
To effectively navigate this complex parameter space, we employ a hybrid two-stage optimization strategy.
 We first use Adam \mbox{
\cite{kingma2014adam}}\hskip0pt
, then refine the phases with L-BFGS when the loss plateaus. Analytical gradient expressions are given in Appendix~\ref{app:grad}.

\begin{figure}[ht!]
 \centering
 \includegraphics[width=\linewidth]{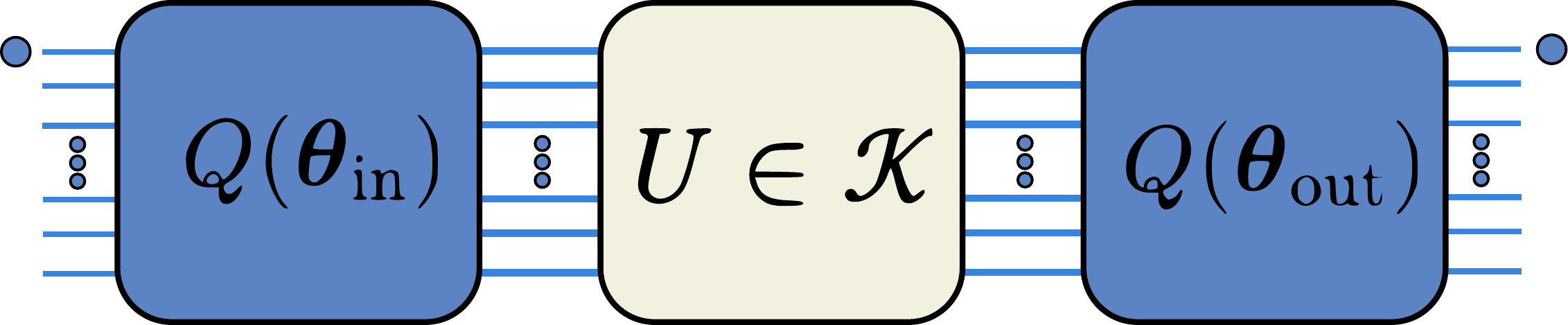}
 \caption{Schematic of the physical test for a 1-design. A single photon sequentially passes through a preparation interferometer $Q(\boldsymbol{\theta}_{\text{in}})$, the unitary operation under test $U \in \mathcal{K}$ from $1$-design set, and a measurement interferometer $Q(\boldsymbol{\theta}_{\text{out}})$.}
 \label{fig:2}
\end{figure}

\section{Resource-efficient interferometers}\label{sec:optimal}
\subsection{Optimal 1-design}
For a mode-space unitary $1$-design, the loss is
 \begin{equation}{\label{eq:7}}
 \mathrm{FPL}(1) = \left(\mathcal{F}^{(1)}_{\mathrm{e.d.}} - 1\right)^2.
\end{equation}

\begin{figure*}[t]
 \centering
 \includegraphics[width=0.85\textwidth]{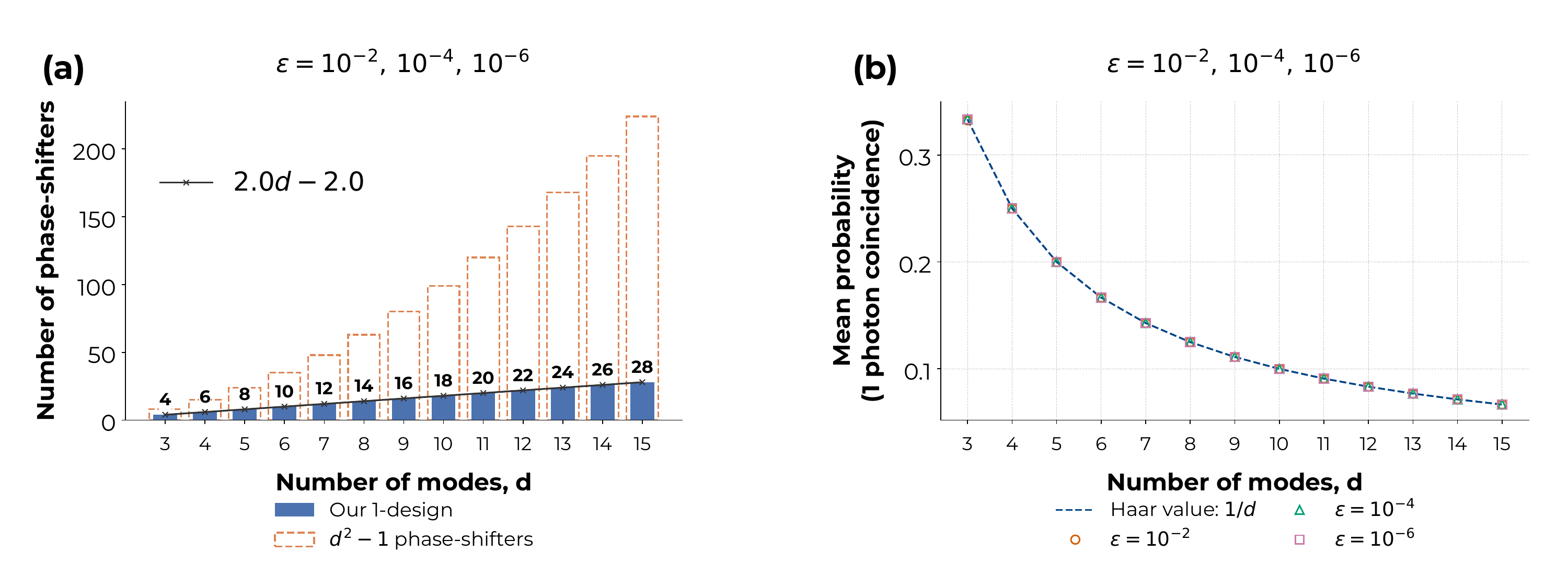}
\caption{Optimal programmable interferometers obtained for $\varepsilon$-approximate unitary $1$-designs at three accuracy thresholds: $\varepsilon = 10^{-2}, 10^{-4}, 10^{-6}$, for $d = 3, \dots, 15$. (a) Number of programmable phase shifters; all three thresholds give $2(d-1)$. (b) Corresponding physical-simulation estimates of the mean transition probability for all three thresholds, compared with the Haar value $1/d$.}
 \label{fig:3}
\end{figure*}

Taking the ensemble size to be $|\mathcal{K}| = d^2$, we perform simulations for several values of $\varepsilon$: $10^{-2}$, $10^{-4}$, and $10^{-6}$.
These values were chosen because they correspond to relative errors of 10\%, 1\%, and 0.1\% with respect to $\mathcal{F}^{(1)}_{\mu_H} = 1$.
We use the single-photon preparation-and-measurement test of Ref.~\cite{matthews2015testing} to assess the recovered first moments.
Suppose a photon is injected into the first mode, and we measure the probability of detecting it in the same mode, averaged over all unitary matrices in the ensemble.
The corresponding transition probability between single-photon Fock states is $P(|1_1, 0, \dots, 0_d\rangle \rightarrow |1_1, 0, \dots, 0_d\rangle) = |U_{11}|^{2}$.
The corresponding Haar average is $\mathbb E_{\mu_H}[|U_{11}|^2]=1/d$, as follows from the first-moment formula in Appendix~\ref{app:wein}. This expresses uniform average occupation of the output modes.

In such an experimental setup, one may place a preparation interferometer $Q(\boldsymbol{\theta}_{\text{in}})$ before $U$ and a measurement interferometer $Q(\boldsymbol{\theta}_{\text{out}})$ after it (see Fig.~\ref{fig:2}).
By tuning the configuration $\boldsymbol{\theta}_{\text{in}}$, an initial single-photon state is prepared in a specific superposition of spatial modes.
Similarly, tuning $\boldsymbol{\theta}_{\text{out}}$ effectively rotates the measurement basis relative to the fixed single-photon detectors.
For a given configuration of $Q(\boldsymbol{\theta}_{\text{in}})$ and $Q(\boldsymbol{\theta}_{\text{out}})$, the measured transition probability corresponds to a specific polynomial -- a linear combination of basis monomials in the matrix elements of $U$ (e.g., terms like $U_{ij} U_{kl}^*$).
To isolate these individual monomials, one cannot rely on a single measurement setting.
Instead, following the approach of Ref.~\cite{matthews2015testing}, one probes the system using a carefully chosen, linearly independent set of configurations.
Measuring the transition probabilities for each configuration in this set generates a full-rank system of linear equations.
By inverting this system, one can algebraically extract the average of any individual monomial up to a desired degree.
For a valid $t$-design, the ensemble average of these extracted monomials over $\mathcal{K}$ must coincide with the theoretical Haar average.
The numerical results for the obtained optimal $\varepsilon$-approximate designs, together with the corresponding physical simulation, are shown in Fig.~\ref{fig:3}.

For every tested dimension $d=3,\ldots,15$, the optimization converges to a two-layer architecture with $2(d-1)$ programmable phase shifters for each of the thresholds $\varepsilon=10^{-2},10^{-4},10^{-6}$, as shown in Fig.~\ref{fig:3}(a).

The mean transition probabilities in Fig.~\ref{fig:3}(b) reproduce the Haar value $1/d$ with negligible residual deviations.

\begin{figure}[ht!]
 \centering
 \includegraphics[width=0.4\textwidth]{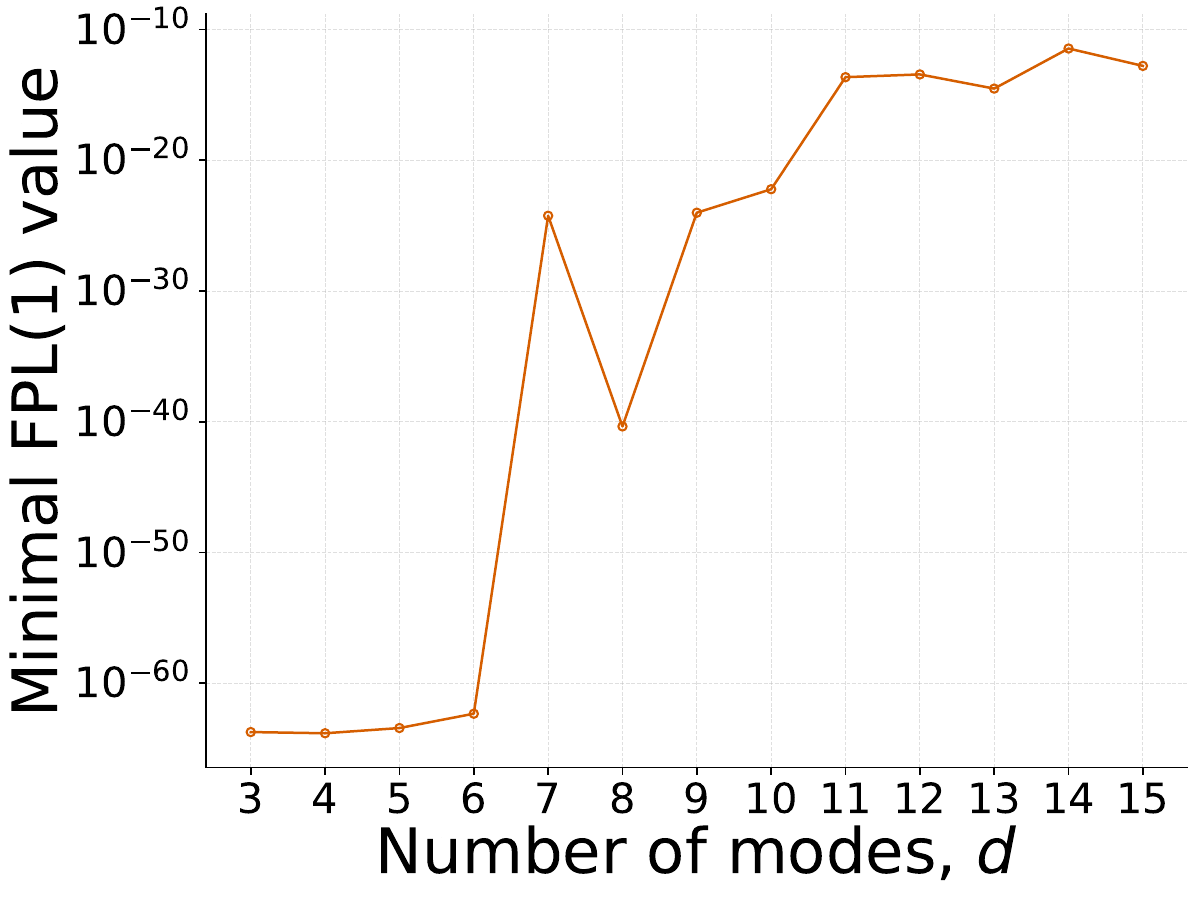}
\caption{Smallest observed frame-potential losses found over $100$ random optimization runs with $L=2$ and $|\mathcal K|=d^2$ for $d=3,\ldots,15$. For $d=13,14,15$, additional basin-hopping continuations are included.}
 \label{fig:4}
\end{figure}

The best-found losses shown in Fig.~\ref{fig:4} confirm that $L=2$ meets the target accuracy throughout the tested range.

In fact, an exact unitary $1$-design can be constructed with $2(d-1)$ programmable phase shifters. This number is minimal within the considered architecture, as established by the theorem below.

\begin{theorem}[Optimal phase-shifter count for a unitary $1$-design]\label{thm:one-design-phases}
    Let $d\ge2$, and let $W\in U(d)$ be a fixed mode-mixing matrix satisfying $|W_{jk}|^2=1/d$ for all $j,k=1,\ldots,d$. In the architecture of Eq.~\eqref{eq:3}, with complete programmable phase layers and one reference phase fixed per layer, the minimum number of programmable phase shifters needed to realize an exact unitary $1$-design is
    \begin{equation}\label{eq:one-design-pmin}
    p_{\text{min}}=2(d-1).
\end{equation}
An ensemble of the minimum cardinality $|\mathcal K|=d^2$ is sufficient. Moreover, within this two-layer layout, fixing any additional individual mode phase prevents the realization of an exact unitary $1$-design, regardless of the ensemble size or the values assigned to the remaining phases.
\end{theorem}

\begin{proof}
We first show that $2(d-1)$ programmable phase shifters are sufficient. Consider the ensemble $\mathcal K=\{U_{ab}\}_{a,b=1}^d$, where
\begin{equation}\label{eq:one-design-uab}
U_{ab}=W\Phi_2(\boldsymbol\theta^{(b)})W\Phi_1(\boldsymbol\theta^{(a)})W,
\qquad a,b=1,\ldots,d.
\end{equation}
Thus the ensemble contains $d^2$ interferometers. With the reference phase fixed to zero, write each phase setting as
\begin{equation}\label{eq:one-design-phase-row}
\boldsymbol\theta^{(a)}=
\left(\theta^{(a)}_1,\ldots,\theta^{(a)}_{d-1},0\right).
\end{equation}
For this construction, collect the phase settings into the matrix
\begin{equation}\label{eq:one-design-theta-matrix}
\boldsymbol\Theta=
\begin{pmatrix}
    \boldsymbol\theta^{(1)}\\
    \boldsymbol\theta^{(2)}\\
    \vdots\\
    \boldsymbol\theta^{(d)}
\end{pmatrix}
=
\begin{pmatrix}
    \theta^{(1)}_1&\cdots&\theta^{(1)}_{d-1}&0\\
    \theta^{(2)}_1&\cdots&\theta^{(2)}_{d-1}&0\\
    \vdots&&\vdots&\vdots\\
    \theta^{(d)}_1&\cdots&\theta^{(d)}_{d-1}&0
\end{pmatrix}.
\end{equation}
Taking the exponential $e^{i(\cdot)}$ of each entry gives
\begin{equation}\label{eq:one-design-h-matrix}
    H=
    \begin{pmatrix}
        e^{i\theta^{(1)}_1}&\cdots&e^{i\theta^{(1)}_{d-1}}&1\\
        e^{i\theta^{(2)}_1}&\cdots&e^{i\theta^{(2)}_{d-1}}&1\\
        \vdots&&\vdots&\vdots\\
        e^{i\theta^{(d)}_1}&\cdots&e^{i\theta^{(d)}_{d-1}}&1
    \end{pmatrix}
    =
    \begin{pmatrix}
        \boldsymbol h^{(1)}\\
        \boldsymbol h^{(2)}\\
        \vdots\\
        \boldsymbol h^{(d)}
    \end{pmatrix}.
\end{equation}
The phase layers are obtained by placing a row $\boldsymbol h^{(a)}$ on the diagonal:
\begin{equation}
    \Phi_l(\boldsymbol\theta^{(a)})=
    \operatorname{diag}\left(h^{(a)}_1,\ldots,h^{(a)}_d\right),
    \qquad l=1,2.
\end{equation}
Choose the rows to be orthogonal, so that
\begin{equation}\label{eq:one-design-hadamard}
    \left\langle\boldsymbol h^{(a)},\boldsymbol h^{(a')}\right\rangle
    =\sum_{k=1}^d\overline{h^{(a)}_k}h^{(a')}_k
    =d\delta_{aa'},\qquad HH^\dagger=d\mathbb I_d.
\end{equation}
For example, for $a,k=1,\ldots,d$, such rows are obtained from
\begin{equation}\label{eq:one-design-phase-basis}
    \theta^{(a)}_k=\frac{2\pi}{d}(a-1)(k-d),\qquad
    h^{(a)}_k=e^{i\theta^{(a)}_k}.
\end{equation}
Indeed, $\theta^{(a)}_d=0$, $h^{(a)}_d=1$, and
\begin{equation}
    \sum_{k=1}^d\overline{h^{(a)}_k}h^{(a')}_k
    =\sum_{k=1}^d\exp\!\left[\frac{2\pi i}{d}(a'-a)(k-d)\right]
    =d\delta_{aa'}.
\end{equation}

Set $V_{ab}=\Phi_2(\boldsymbol\theta^{(b)})W\Phi_1(\boldsymbol\theta^{(a)})$, so that $U_{ab}=WV_{ab}W$. By unitarity of $W$ and cyclicity of the trace,
\begin{multline}
    \operatorname{Tr}(U_{ab}^\dagger U_{a'b'})=
    \operatorname{Tr}\left(W^\dagger V_{ab}^\dagger W^\dagger W V_{a'b'}W\right)=\\
    \operatorname{Tr}\left(V_{ab}^\dagger V_{a'b'}\right).
\end{multline}
Since $(V_{ab})_{jk}=h^{(b)}_jW_{jk}h^{(a)}_k$, the flat-modulus condition $|W_{jk}|^2=1/d$ gives
\begin{multline}\label{eq:one-design-unitary-basis}
    \operatorname{Tr}(U_{ab}^\dagger U_{a'b'})=
    \sum_{j,k=1}^d\overline{(V_{ab})_{jk}}(V_{a'b'})_{jk}=\\
    \sum_{j,k=1}^d\overline{h^{(b)}_j}h^{(b')}_j
    \overline{h^{(a)}_k}h^{(a')}_k|W_{jk}|^2=\\
    \frac1d\left(\sum_j\overline{h^{(b)}_j}h^{(b')}_j\right)
    \left(\sum_k\overline{h^{(a)}_k}h^{(a')}_k\right)=\\
    \frac1d\left\langle\boldsymbol h^{(b)},\boldsymbol h^{(b')}\right\rangle
    \left\langle\boldsymbol h^{(a)},\boldsymbol h^{(a')}\right\rangle
    =d\delta_{aa'}\delta_{bb'}.
\end{multline}
Substituting this into the frame potential, we obtain
\begin{equation}\label{eq:one-design-exact-fp}
    \mathcal F^{(1)}_{\mathrm{e.d.}}
    =\frac1{d^4}\sum_{a,b,a',b'}
    \left|\operatorname{Tr}(U_{ab}^\dagger U_{a'b'})\right|^2
    =\frac{d^2\cdot d^2}{d^4}=1.
\end{equation}
Thus two phase layers with $2(d-1)$ programmable phase shifters realize an exact unitary $1$-design.

A single phase layer cannot realize such a design. Indeed, an ensemble of $U_i=W\Phi_iW$ would be a $1$-design if and only if the diagonal matrices $W^\dagger U_iW^\dagger=\Phi_i$ were a $1$-design. For $r\ne k$, their matrix elements obey $|(\Phi_i)_{rk}|^2=0$, whereas a $1$-design requires the average $1/d$. Thus $L\ge2$. Together with the construction above and $p=L(d-1)$, this proves Eq.~\eqref{eq:one-design-pmin} for the complete-layer architecture.

We finally show that no additional individual phase may be fixed in the two-layer layout. For an arbitrary member of this family, let
\begin{equation}
    V=\Phi_2W\Phi_1,\qquad
    V_{rk}=e^{i\theta_{2,r}}W_{rk}e^{i\theta_{1,k}},
\end{equation}
where $\theta_{1,d}=\theta_{2,d}=0$. If the original ensemble is a $1$-design, the ensemble of $V$ matrices is also a $1$-design, so
\begin{equation}\label{eq:one-design-first-moment}
    \frac1{|\mathcal K|}\sum_i (V_i)_{rk}\overline{(V_i)_{sl}}
    =\frac{\delta_{rs}\delta_{kl}}d.
\end{equation}
Suppose that an additional phase in the first layer is fixed, $\theta_{1,j}=c$ with $j\ne d$. Then
\begin{equation}
    V_{rj}\overline{V_{rd}}
    =e^{ic}W_{rj}\overline{W_{rd}},\qquad
    \left|e^{ic}W_{rj}\overline{W_{rd}}\right|=\frac1d.
\end{equation}
This product is independent of every remaining programmable phase. Its ensemble average is therefore nonzero, in contradiction with Eq.~\eqref{eq:one-design-first-moment}, which requires zero for $j\ne d$.

If instead $\theta_{2,j}=c$ is fixed in the second layer, the same contradiction follows from two entries of one column:
\begin{equation}
    V_{jr}\overline{V_{dr}}=e^{ic}W_{jr}\overline{W_{dr}}\ne0.
\end{equation}
Hence all $2(d-1)$ non-reference mode phases must remain programmable in this two-layer layout.
\end{proof}

To understand why a two-layer architecture ($L=2$) serves as a critical threshold for generating a unitary 1-design, it is instructive to consider the degrees of freedom required to parameterize a single-photon state.
Mathematically, the action of a unitary matrix $U$ on a single input photon corresponds to a single column of the matrix, which represents a normalized complex vector in $d$-dimensional space.
Defining an arbitrary complex state vector up to normalization and a global phase requires exactly $2d-2$ independent real parameters.
A single programmable layer ($L=1$) embedded between static DFT matrices provides only $d-1$ independent phase parameters, severely restricting the accessible single-particle states to a lower-dimensional subspace.
The introduction of the second layer ($L=2$) expands the parameter space of each individual matrix to exactly $2(d-1) = 2d-2$ degrees of freedom.
This value perfectly matches the dimensionality of the single-photon state space.
Physically, interleaving two programmable phase layers with balanced DFT mode-mixing provides the individual interferometer with the minimal expressive power necessary to map any input mode to an arbitrary coherent superposition at the output.
This capability is a fundamental prerequisite for the generated ensemble to uniformly cover the single-particle Hilbert space and perfectly satisfy the 1-design moment constraints without the need for a full quadratic $\mathcal{O}(d^2)$ universal parametrization.

This conclusion aligns with the results of Ref.~\cite{zelaya2025integrated}, which demonstrated that two random phase layers are sufficient to produce output optical signals with a white-noise profile, even for highly sparse input signals. In optical networks, such behavior reflects strong statistical mixing and redistribution of optical intensity across the output modes. This is closely related to the uniform ensemble-averaged intensities required to reproduce the first moment of the Haar measure. In this sense, a two-layer interferometer possesses the structural complexity needed to generate a broad class of output states.

Frame-potential minimization provides a way to search for ensembles subject to fixed hardware constraints. Appendix~\ref{app:dif} compares one numerical single-photon ensemble with standard qutrit Weyl operators. Visual differences do not establish inequivalence under fixed left and right unitary transformations, phases, and permutations of ensemble elements.

\subsection{Two-photon designs}

\begin{figure*}[t]
    \centering
    \includegraphics[width=0.85\textwidth]{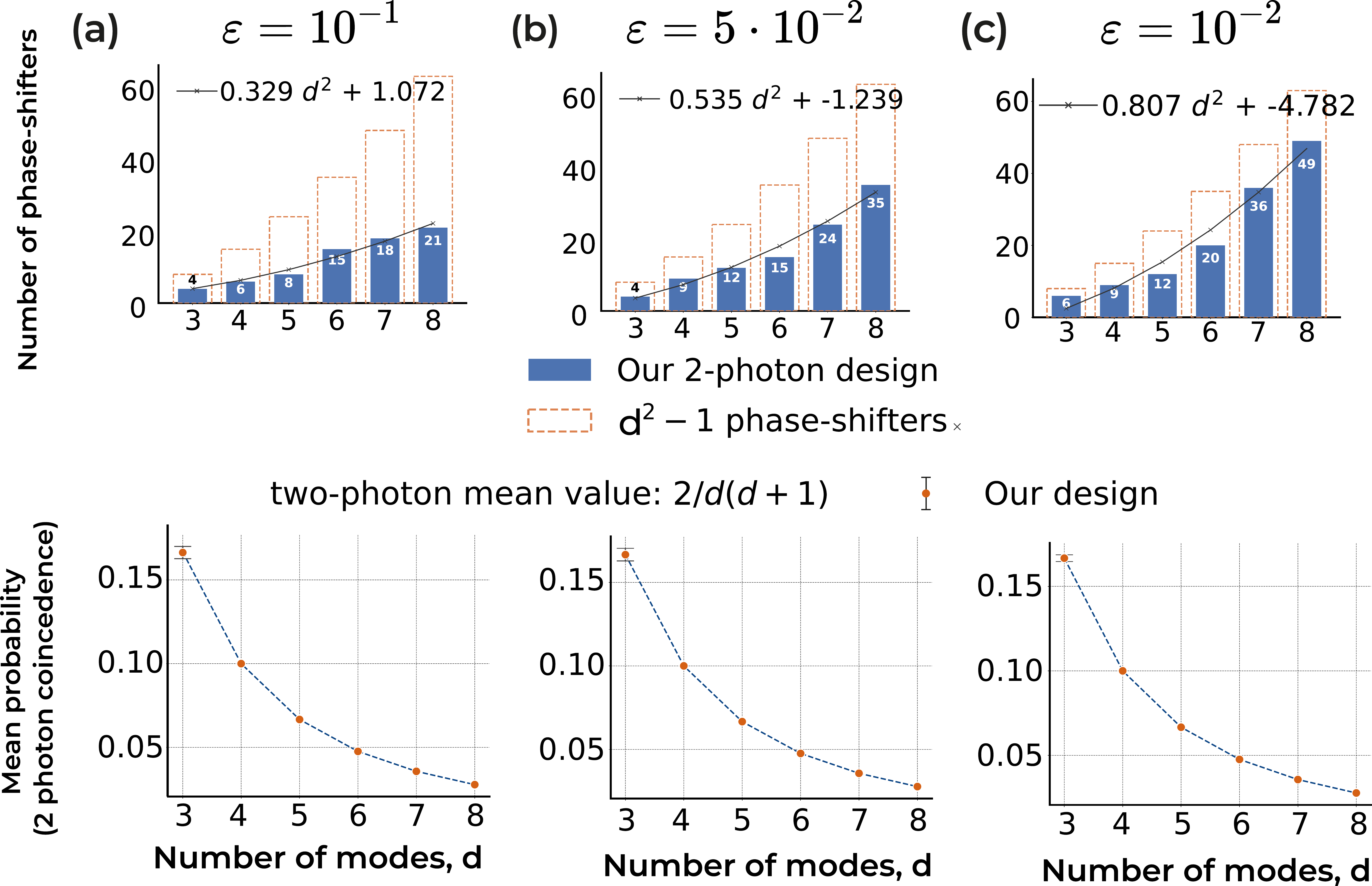}
\caption{Numerically selected interferometers for the two-photon loss $\mathrm{FPL}_{2\mathrm{ph}}(1)$ at thresholds (a) $10^{-1}$, (b) $5\cdot10^{-2}$, and (c) $10^{-2}$, for $d=3,\ldots,8$. The lower panels show simulated two-photon transition probabilities.}
    \label{fig:5}
\end{figure*}

In this case, Eq.~\eqref{eq:5} is no longer sufficient.
The reason is that, experimentally, one can directly access only probabilities, which are of the form $|\cdot|^2$ in the matrix elements of $U$.
Thus, in the $1$-design case, where we considered single-photon evolution under $U \in U(d)$, one does not obtain an observable whose average reproduces higher Haar moments for $t > 1$, because only one matrix element might be in $|\cdot|^2$ expression.
To access these higher moments, one has to consider multi-photon interference, and the simplest nontrivial setting is the two-photon case.
Two indistinguishable photons evolve in the symmetric space $\mathrm{Sym}^2(\mathbb C^d)$ of dimension $D_2=d(d+1)/2$. We denote the restriction of $U\otimes U$ to this space by $R_2(U)$. The evolution can be represented conveniently in the ambient tensor-product space in first-quantized form as

\begin{equation}{\label{eq:8}}
    \mathcal U^{(2)}=(U\otimes U)P_{\mathrm{sym}},
\end{equation}
\noindent where $P_{\mathrm{sym}}$ is the projector onto the symmetric subspace occupied by indistinguishable bosons. For two photons,
\begin{equation}{\label{eq:9}}
    P_{\mathrm{sym}}=\frac12(\mathbb I+\mathbb F),
\end{equation}
where $\mathbb I$ denotes the identity operator and $\mathbb F$ denotes the swap operator. Applying the frame-potential definition to $R_2$ gives

\begin{equation}{\label{eq:10}}
    \mathcal{F}^{(t),2\mathrm{ph}}_{\mu_H} = \frac{1}{4^{t}} \int_{U(d)} \left|\left( \mathrm{Tr} U\right)^2 + \mathrm{Tr} U^2 \right|^{2t} d\mu_H(U),
\end{equation}

\begin{equation}{\label{eq:11}}
    \mathcal{F}^{(t),2\mathrm{ph}}_{\mathrm{e.d.}} = \frac{1}{4^t |\mathcal{K}|^2} \sum_{i,j = 1}^{|\mathcal{K}|} \left| \left( \mathrm{Tr}(U_i U_j^{\dagger}) \right)^2 + \mathrm{Tr}\left( (U_i U_j^{\dagger})^2 \right) \right|^{2t}.
\end{equation}
These expressions can be derived using representation theory and, in particular, the character formalism discussed in Appendix~\ref{app:repr}. Here, $t$ denotes the frame-potential order in the two-photon representation $R_2$. At $t=1$, the polynomial degree in the entries of the mode matrix is determined by the photon number $n$: it is $n$ in the entries of $U$ and $n$ in their complex conjugates, giving a total degree of $2n$, as in the $n$-photon transition probabilities discussed in Ref.~\cite{matthews2015testing}. Thus, in the present two-photon setting, $t=1$ probes polynomials of balanced degree $(2,2)$, i.e., total degree $4$. For general $t$, the two-photon Haar integrand in Eq.~\eqref{eq:10} has balanced degree $(2t,2t)$.

We optimize $t=1$ in this representation. Since $R_2$ is irreducible, its Haar twirling channel is depolarizing. An exact two-photon design satisfies, for every operator $X$ on $\mathrm{Sym}^2(\mathbb C^d)$,
\begin{equation}\label{eq:bosonic-design}
\frac1{|\mathcal K|}\sum_i R_2(U_i)X R_2(U_i)^\dagger
=\frac{\operatorname{Tr}X}{D_2}\mathbb I_{D_2}.
\end{equation}

Equivalently, the induced matrices form a unitary $1$-design on $U(D_2)$, and their first frame potential is $1$. This guarantees the Haar averages of all two-photon transition probabilities. Throughout this section, the design condition is imposed on the two-photon space $\mathrm{Sym}^2(\mathbb C^d)$ through the induced transformation $R_2(U)$.

For consistency with the existing calculations, we use $|\mathcal K|=d^4-2d^2+2$ as a chosen ensemble size. The loss used in the simulations is

\begin{equation}{\label{eq:12}}
    \mathrm{FPL}_{2\mathrm{ph}}(1) = \left(\mathcal{F}^{(1), 2\mathrm{ph}}_{\mathrm{e.d.}} - 1 \right)^2.
\end{equation}

\begin{figure*}[t]
    \centering
    \includegraphics[width=\textwidth]{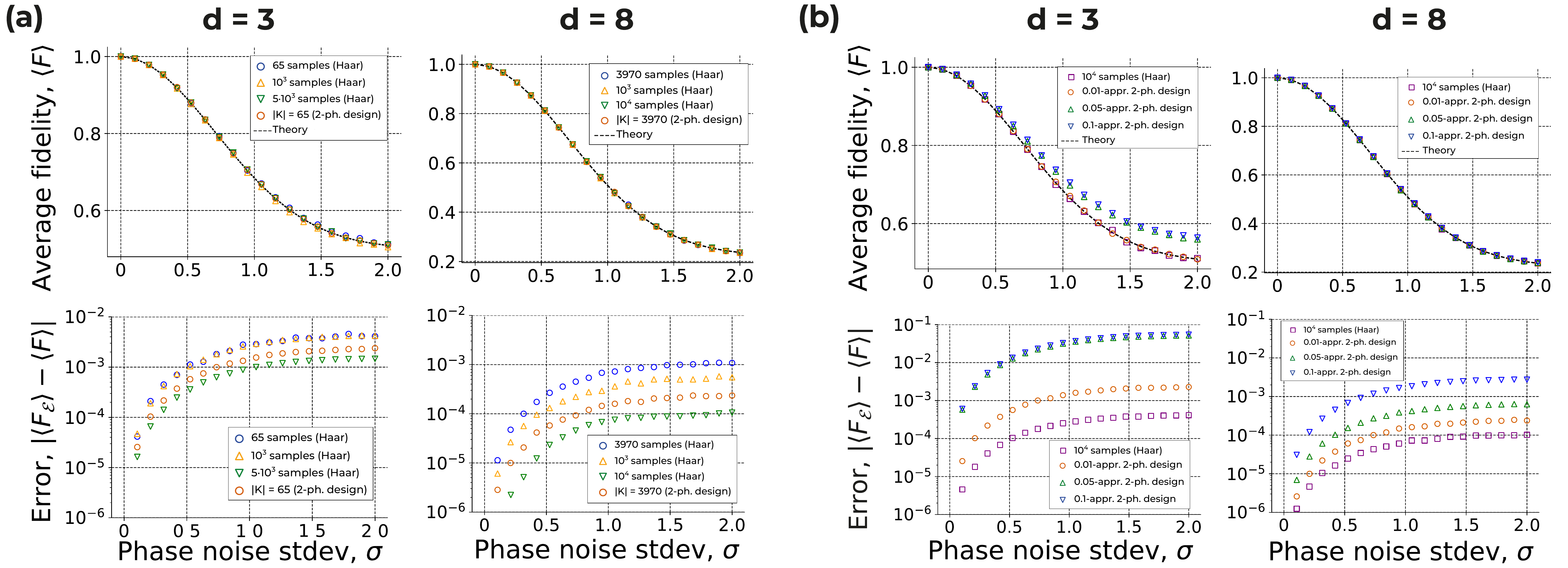}
\caption{Average-fidelity estimation for the noisy channel in Eq.~ \eqref{eq:13}. (a) Analytical Haar fidelity, finite Haar-sampling estimates, and estimates from the two-photon ensemble with loss threshold $0.01$, for $d=3$ and $8$, versus phase-noise standard deviation $\sigma$. (b) Errors for the two-photon ensembles of Fig.~\ref{fig:5}.}
    \label{fig:6}
\end{figure*}

The physical simulation protocol is analogous to the one used for the 1-design case, except that we now consider two photons injected into the first and second input modes and measure the probability of detecting them in the same output modes.
Without loss of generality, restricting our attention to these specific modes is sufficient, as arbitrary mode permutations can be straightforwardly implemented in linear optical networks.
The corresponding transition probability for a two-photon state is $P(|1_1, 1_2, 0, \dots, 0_d\rangle \rightarrow |1_1, 1_2, 0, \dots, 0_d\rangle) = |U_{11}U_{22}+U_{12}U_{21}|^{2}$, while its Haar average is $\mathbb{E}_{\mu_H}[|U_{11}U_{22}+U_{12}U_{21}|^2] = \frac{2}{d(d+1)}$, which can be understood as finding one possible outcome over $\binom{d+1}{2}$.
The optimization and physical-simulation results for different values of $\varepsilon$ are presented in Fig.~\ref{fig:5}.

This transition probability provides an experimentally accessible benchmark for the optimized two-photon ensemble.

As illustrated in the lower panels of Fig.~\ref{fig:5}, the optimized numerical ensembles successfully pass the physical simulation test for two-photon interference.
Across all evaluated dimensions and accuracy thresholds, the extracted mean probability of a two-photon coincidence perfectly aligns with the theoretical Haar expectation of $2/[d(d+1)]$.
The negligible error bars confirm that the generated $\varepsilon$-approximate two-photon designs consistently and reliably reproduce second-order polynomial observables, demonstrating that the structural complexity of these shallow interferometers is sufficient to mimic Haar-random multi-photon statistics.

Furthermore, the generated $\varepsilon$-approximate two-photon designs are structurally unique when compared to standard qubit-based constructions.
Conventional discrete $2$-designs, such as the multi-qubit Clifford group, rely on deterministic two-qubit entangling gates.
However, such deterministic non-linear operations are fundamentally impossible to implement in linear-optical networks without measurement-induced nonlinearities or auxiliary resources.
Instead, our frame-potential optimization algorithm discovers specialized phase configurations that leverage passive multi-photon interference (boson bunching) directly in the Fock space to satisfy the two-photon design constraints with prescribed error value $\varepsilon$.
This generative approach yields novel matrix ensembles that are distinct from known algebraic groups and are tailored to the physical constraints of linear optics.

These results show that constructing an interferometer capable of realizing a two-photon design is significantly more demanding than in the $1$-design case.
Each threshold value of $\varepsilon$ in Fig.~\ref{fig:5} imposes a corresponding lower requirement on the number of programmable parameters.
As expected, achieving a more accurate approximation to a two-photon design requires more tunable parameters.
Nevertheless, we are still able to construct $\varepsilon$-approximate two-photon designs with high accuracy, as confirmed by the small errors in the mean-probability test.

\section{Benchmark}\label{sec:benchmark}

\subsection{Average fidelity of a noisy channel}
A possible application of unitary $t$-designs in quantum information processing is the twirling of quantum channels \cite{nielsen2002simple, dankert2009exact, OliviaDiMatteo2021, lancien2020weak}.
Suppose we have a quantum channel $\mathcal{E}(\rho)$ that ideally implements a target unitary transformation $V$, but suffers from an induced phase shift on all modes.
Each phase error $\phi_k$ is an independent and identically distributed (i.i.d.) random variable following a normal distribution with zero mean and variance $\sigma^2$: $P(\boldsymbol{\phi}) = \prod_{k=1}^d P(\phi_k) = \prod_{k=1}^d \frac{1}{\sqrt{2\pi\sigma^2}} e^{-\frac{\phi_k^2}{2\sigma^2}}$.
The noisy channel is

\begin{equation}{\label{eq:13}}
    \mathcal{E}(\rho) = \int d\boldsymbol{\phi} P(\boldsymbol{\phi}) N(\boldsymbol{\phi}) V \rho V^\dagger N(\boldsymbol{\phi})^\dagger
\end{equation}
where $N(\boldsymbol{\phi}) = \mathrm{diag}\left(e^{i\phi_1}, e^{i\phi_2}, \dots, e^{i\phi_d}\right)$. The average channel fidelity is analytically given by

\begin{multline}{\label{eq:14}}
    \langle F_{\mathcal{E}}\rangle = \int_{U(d)} \bra{0} U^{\dagger} V^{\dagger}\ \mathcal{E} \left(U \ket{0} \bra{0} U^{\dagger}  \right) V U \ket{0}\ d\mu_H = \\
    =  \frac{2 + (d-1)e^{-\sigma^2}}{d+1}
\end{multline}
where this result is obtained using Weingarten calculus and derived in detail in Appendix~\ref{app:twirl}.

Since Eq.~\eqref{eq:14} contains two copies of $U$ and two copies of $U^\dagger$, a full mode-space unitary $2$-design is sufficient to replace the Haar integral by an ensemble average.
For this observable, the weaker two-photon condition in Eq.~\eqref{eq:bosonic-design} also suffices.
We estimate the fidelity with the optimized approximate ensembles as

\begin{equation}{\label{eq:15}}
    \langle F \rangle = \frac{1}{|\mathcal{K}|} \sum_{i = 1}^{|\mathcal{K}|}\ \bra{0} U_{i}^{\dagger} V^{\dagger} \mathcal{E} \left(U_i \ket{0} \bra{0} U_i^{\dagger} \right) VU_{i} \ket{0}.
\end{equation}

To numerically evaluate the action of the noisy channel \eqref{eq:13} during the simulation, the continuous integral over the phase errors $\boldsymbol{\phi}$ is approximated using Monte Carlo sampling, wherein the fidelity is averaged over $N = 10\,000$ independent random realizations of the phase noise vector.

Figure~\ref{fig:6}(a) compares the analytical Haar-averaged fidelity with estimates obtained from the optimized approximate two-photon designs at loss threshold $0.01$.
We plot the average fidelity and the absolute estimation error as a function of the phase noise standard deviation $\sigma$ for $d=3$ and $d=8$.
The ensemble estimate in Eq.~\eqref{eq:15} is compared directly with the analytical fidelity in Eq.~\eqref{eq:14}.
Haar-based estimates with several sample sizes provide finite-sampling baselines. In particular, we include a random Haar ensemble of the same cardinality as each generated two-photon design: $|\mathcal{K}|=65$ for $d=3$ and $|\mathcal{K}|=3970$ for $d=8$.
In the displayed benchmark, the optimized approximate two-photon design consistently outperforms a random Haar ensemble of the same cardinality, providing more accurate fidelity estimates.

This advantage of utilizing structured $t$-designs over unstructured Haar sampling strongly resonates with recent advancements in efficient quantum state learning, most notably the framework of classical shadows \cite{huang2020predicting}.
In the classical shadow formalism, randomized measurements drawn from discrete unitary $t$-designs (such as the multi-qubit Clifford group, which forms a 3-design) are employed to construct highly sample-efficient estimators for global quantum properties, including fidelity.
The theoretical power of classical shadows stems from the fact that $t$-designs guarantee a uniform coverage of the relevant operator space, strictly bounding the variance of the estimators.
Our numerical results demonstrate a continuous-variable optical analogue of this phenomenon: by numerically forcing the shallow interferometer ensemble to satisfy the exact two-photon design moment constraints, we effectively "derandomize" the channel twirling process.
This structured approach drastically suppresses the statistical fluctuations inherent to pure random Haar sampling, proving that hardware-efficient approximate two-photon designs can serve as powerful tools for sample-optimal fidelity estimation in linear optics.

Figure~\ref{fig:6}(b) illustrates the impact of the approximation parameter $\varepsilon$ by comparing a dense baseline of $10^4$ Haar samples against the family of $\varepsilon$-approximate two-photon designs ($\varepsilon = 0.01, 0.05, 0.1$) introduced in Fig.~\ref{fig:5}.
A clear hierarchy is observed: smaller values of $\varepsilon$ systematically lead to more accurate fidelity estimations, with the error saturating at a characteristic floor dictated by the approximation precision.
This behavior confirms that interferometers with more programmable parameters can be optimized to approximate the target two-photon design more closely, thus resolving the Haar-averaged fidelity with higher accuracy up to the fundamental limit imposed by $\varepsilon$.

\begin{figure}[t]
    \centering
    \includegraphics[width=\linewidth]{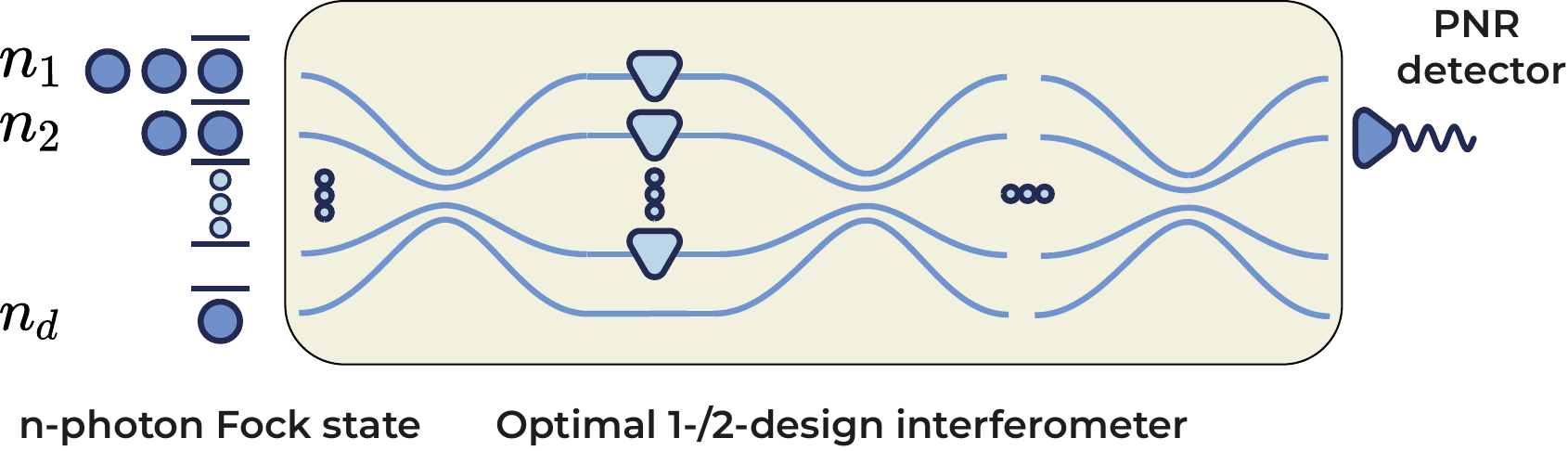}
    \caption{Benchmark for recovering the total photon number of a multi-photon Fock state $\ket{\psi} = \ket{n_1, n_2, \dots, n_d}$, where $\sum_{i =1}^{d}n_i = n$.}
    \label{fig:get_photons}
\end{figure}

\begin{figure}[t]
    \centering
    \includegraphics[width=\linewidth]{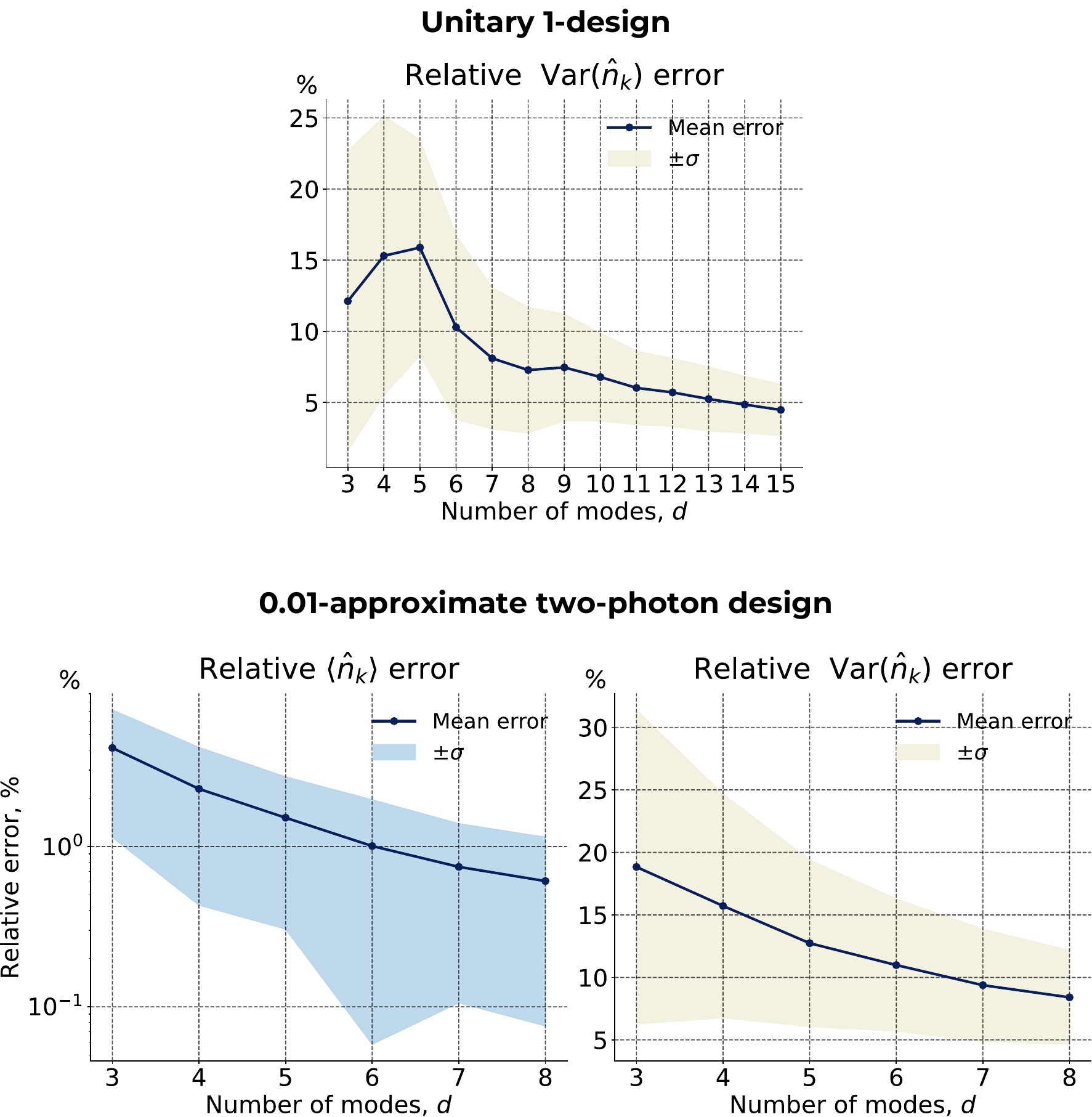}
    \caption{Relative errors in photon-number estimation. The upper panel shows the variance error obtained with the optimal unitary $1$-design for the lowest obtained $\varepsilon$ values from Fig.~\ref{fig:4}. The lower panels show the mean-occupation error (left) and variance error (right) for the $0.01$-approximate two-photon design.}
    \label{fig:rel_err}
\end{figure}

\subsection{Recovering the number of photons of multi-photon Fock state}

A novel practical problem that can be solved by the constructed optimal interferometers is the recovery of the total photon number using a single monitored output mode.
As illustrated in Fig.~\ref{fig:get_photons}, we consider an input state $\ket{\psi} = \ket{n_1, n_2, \dots, n_d}$ with total photon number $n = \sum_{i=1}^{d} n_i$, send it through an interferometer drawn from the ensemble $\mathcal{K}$, and measure the occupation number in a single output mode $k$ with a photon-number-resolving (PNR) detector.
Averaging over a mode-space unitary $1$-design gives the Haar-predicted first moment,

\begin{equation}{\label{eq:16}}
    \mathbb{E}_{\mu_H}[\langle \hat{n}_k\rangle ]  = \frac{n}{d}.
\end{equation}

To characterize the statistical fluctuations of this estimator, we consider the variance of photon counts when both the interferometer $U$ and the detection outcome are sampled. We denote this variance by $\mathrm{Var}_{\mu_H}(\hat n_k)$, while $\langle\cdot\rangle_U$ denotes the quantum expectation for a fixed $U$:

\begin{multline}\label{eq:17}
    \mathrm{Var}_{\mu_H}(\hat n_k)
    =\mathbb{E}_{\mu_H}\left[\left\langle\left(\hat n_k-\frac nd\right)^2\right\rangle_U\right] = \\
    =\mathbb{E}_{\mu_H}\left[\langle\hat n_k^2\rangle_U\right]
    -\left(\mathbb{E}_{\mu_H}\left[\langle\hat n_k\rangle_U\right]\right)^2
    =\frac{n(d-1)(n+d)}{d^2(d+1)},
\end{multline}

which can be reproduced by averaging over a unitary $2$-design ensemble $\mathcal{K}$; the derivation is given in Appendix~\ref{app:estim}.

The numerical results for this benchmark are presented in Fig.~\ref{fig:rel_err}. For each value of $d$, we perform the following procedure:
\begin{enumerate}
    \item Generate $10^3$ random Fock states with total photon number chosen uniformly from $1$ to $d$. Although the analytical results in Eqs.~\eqref{eq:16} and \eqref{eq:17} hold for arbitrarily large $n$ (including $n > d$), this specific range is chosen to maintain computational tractability.
    \item Evaluate the estimated mean occupation number and its variance for a selected output mode $k$ monitored by an ideal PNR detector capable of resolving an arbitrary number of photons.
    \item Compare the results with the theoretical predictions in Eqs.~\eqref{eq:16} and \eqref{eq:17}.
\end{enumerate}

The relative errors shown in Fig.~\ref{fig:rel_err} quantify how accurately the constructed ensembles reproduce the Haar-averaged first and second moments relevant to photon-number estimation.
For the exact unitary $1$-design, the mean occupation number equals $n/d$, so its relative error vanishes and the corresponding panel is omitted.
In contrast, the error in the variance remains significantly larger, showing that a unitary $1$-design is not sufficient to reproduce photon-number fluctuations, which depend on second moments of the ensemble.

The $0.01$-approximate two-photon design also does not yet provide the expected improvement in this benchmark and, in our simulations, leads to even larger relative errors. This suggests that photon-number estimation is a particularly demanding test of second-moment accuracy.

\section{Conclusion}\label{sec:discussion}

We studied the construction of unitary $1$-designs and two-photon designs on programmable linear-optical interferometers under the constraint of minimizing the number of tunable phase shifters.
Using the frame-potential loss as the optimization criterion, we identified interferometer architectures that realize unitary $1$-designs and approximate two-photon designs with substantially fewer programmable parameters than are required to implement a universal transformation in $SU(d)$, which requires $d^2-1$ independent phase-shifter parameters.

For the considered interferometer architecture, we established the exact minimum number of programmable phase shifters required to realize a unitary $1$-design, $p_{\mathrm{min}}=2(d-1)$, and provided an explicit construction attaining this bound. The numerical optimization results confirm this parameter count.

The ensembles were benchmarked using transition probabilities, average channel fidelity, and photon-number estimation.
Approximate two-photon designs reproduced Haar-averaged fidelities with good accuracy, showing that ensemble averages over a relatively small number of optimized matrices can outperform direct randomized sampling over large Haar-random ensembles.
At the same time, the novel photon-number-estimation experiment revealed a more stringent requirement: although a unitary $1$-design already reproduces first-moment observables well, accurately capturing photon-number fluctuations requires much better control of second moments.
This makes photon-number estimation a sensitive test of approximate two-photon design quality.

These results motivate optical moment designs tailored to specified experimental observables. A separate goal is to construct full unitary $2$-designs on a multiphoton Fock space. Passive linear optics alone cannot meet that exact condition; additional nonlinear operations, such as the SNAP gates discussed in Ref.~\cite{saied2024advancing}, provide a route toward universality. Proposals for photon--photon gates based on quantum emitters offer another relevant direction \cite{pettersson2026high}.

Finally, an important consideration for the experimental implementation of such shallow architectures is their significant reduction in optical propagation losses.
While minimizing the circuit depth reduces propagation losses and dynamic control complexity, it simultaneously removes the structural redundancy inherent in universal $\mathcal{O}(d^2)$ networks.
Consequently, systematic errors in the interferometer -- such as those arising from thermal crosstalk or imperfect static calibration of individual phase shifters -- will directly perturb the generated unitary matrices.
Such systematic deviations will inherently increase the frame potential loss, degrading the precision of the achieved approximate $t$-design.
Understanding the exact sensitivity of the frame potential to static hardware noise, and exploring the trade-off between circuit depth (redundancy) and resilience to calibration errors, represents a critical direction for future experimental characterization.

\section{Acknowledgments}

The work was supported by Rosatom in the framework of the Roadmap for Quantum computing (Contract № 868/1759-D dated 3 October 2025 and Contract № P2529 dated 29 December 2025).

\bibliography{reference}

\onecolumngrid
\appendix
\newpage
\section{Computing gradients of FPL$(t)$}\label{app:grad}
In \cite{kuzmin2021architecture}, exact gradient formulas were derived using a convenient factorized representation of $U_i$:

\begin{equation}\label{app_eq:1}
    U_i = A^{(i)}_{l} \Phi^{(i)}_{l} B^{(i)}_l,
\end{equation}
where

\begin{minipage}{0.45\textwidth}
      \begin{equation*}
        \begin{cases}
            A^{(i)}_{L} = W, \\
            A^{(i)}_{k} = A^{(i)}_{k+1}\, \Phi^{(i)}_{k+1}\, W,\\
            k = \overline{L-1, 1},
        \end{cases}
    \end{equation*}
\end{minipage}
\hfill
\begin{minipage}{0.45\textwidth}
    \begin{equation}\label{app_eq:3}
        \begin{cases}
            B^{(i)}_{1} = W, \\
            B^{(i)}_{k} = W\, \Phi^{(i)}_{k-1} B^{(i)}_{k-1},\\
            k = \overline{2, L}.
        \end{cases}
    \end{equation}
\end{minipage}

The gradient of $U_i$ with respect to the phase parameter $\theta^{(i)}_{l,m}$ is given by
\begin{equation}\label{app_eq:4}
    \frac{\partial U^{(i)}}{\partial \theta^{(i)}_{l,m}} = A^{(i)}_{l}\, \frac{\partial \Phi^{(i)}_l}{\partial \theta^{(i)}_{l,m}}\, B^{(i)}_{l}
    = A^{(i)}_{l}\, i e^{i\theta^{(i)}_{l,m}} \ket{m}\bra{m}\, B^{(i)}_{l} = i e^{i\theta^{(i)}_{l,m}}\, A^{(i)}_{l}\, \ket{m}\bra{m}\, B^{(i)}_{l},
\end{equation}
where $\ket{m}\bra{m}$ denotes a matrix whose only nonzero entry is the $(m,m)$ element, equal to $1$. The general expression for the gradient of the frame potential loss is
\begin{equation}\label{app_eq:5}
    \frac{\partial \mathrm{FPL}(t)}{\partial \theta^{(i)}_{l,m}} = 2 \left(\mathcal{F}^{(t)}_{\mathrm{e.d.}} - \mathcal{F}^{(t)}_{\mu_H}\right) \frac{\partial \mathcal{F}^{(t)}_{\mathrm{e.d.}}}{\partial \theta^{(i)}_{l,m}}.
\end{equation}

We first derive the expression for a unitary $1$-design. Define $G_{ij} = U_i U_j^{\dagger}$. Then
\begin{equation}\label{app_eq:6}
    \frac{\partial \mathcal{F}^{(1)}_{\mathrm{e.d.}}}{\partial \theta^{(i)}_{l,m}} = \frac{1}{|\mathcal{K}|^2} \sum_{j = 1}^{|\mathcal{K}|}
    \left( \frac{\partial \left|\mathrm{Tr} G_{ij}\right|^2}{\partial \theta^{(i)}_{l,m}} + \frac{\partial \left|\mathrm{Tr} G_{ji}\right|^2}{\partial \theta^{(i)}_{l,m}} \right) \\
    = \frac{4}{|\mathcal{K}|^2} \sum_{j = 1}^{|\mathcal{K}|} \Re \left[ \mathrm{Tr} G_{ji} \frac{\partial \mathrm{Tr} G_{ij}}{\partial \theta^{(i)}_{l,m}} \right],
\end{equation}
where we used the identity $\frac{\partial |z(x,y)|^2}{\partial x} = 2\Re\!\left(z^{\ast}\frac{\partial z}{\partial x}\right)$ for any complex-valued function $z$. Furthermore,
\begin{equation}\label{app_eq:7}
    \frac{\partial \mathrm{Tr} G_{ij}}{\partial \theta^{(i)}_{l,m}} = \mathrm{Tr} \left( \frac{\partial U_i}{\partial \theta^{(i)}_{l,m}} U_j^{\dagger} \right)
    = i e^{i\theta^{(i)}_{l,m}} \mathrm{Tr} \left(A^{(i)}_l \ket{m}\bra{m} B^{(i)}_{l} U_j^{\dagger}\right)
    = i e^{i\theta^{(i)}_{l,m}} \bra{m} B^{(i)}_l U_j^{\dagger} A^{(i)}_l \ket{m}.
\end{equation}

Using   $\operatorname{Re}(iz)=-\operatorname{Im}z$, Eqs.~\eqref{app_eq:5}--\eqref{app_eq:7}, and $M^{(1)}_{ij}=\operatorname{Tr}(U_jU_i^\dagger)U_j^\dagger$ , we obtain
  \begin{equation}\label{app_eq:8}
\frac{\partial\mathrm{FPL}(1)}{\partial\theta^{(i)}_{l,m}}
=-\frac{8(\mathcal F^{(1)}_{\mathrm{e.d.}}-1)}{|\mathcal K|^2}
\sum_j\operatorname{Im}\left[e^{i\theta^{(i)}_{l,m}}
\langle m|B_l^{(i)}M_{ij}^{(1)}A_l^{(i)}|m\rangle\right].
\end{equation}

  The same calculation applies to the first frame potential of the  two-photon   representation. The character derivative is
\begin{multline}\label{app_eq:9}
    \frac{\partial}{\partial \theta^{(i)}_{l,m}}\frac{1}{2} \left[(\mathrm{Tr} G_{ij})^2 + \mathrm{Tr} G_{ij}^{2} \right]
    = \mathrm{Tr} G_{ij}\, \mathrm{Tr} \frac{\partial G_{ij}}{\partial \theta^{(i)}_{l,m}} + \mathrm{Tr} \left( G_{ij} \frac{\partial G_{ij}}{\partial \theta^{(i)}_{l,m}} \right) \\
    = i e^{i \theta^{(i)}_{l,m}} \bra{m} B^{(i)}_{l} \left[ \mathrm{Tr}(G_{ij})\mathbb I + U_{j}^{\dagger} U_i \right] U_j^{\dagger} A^{(i)}_{l}\ket{m}.
\end{multline}

  Applying the chain rule to the loss in Eq.~\eqref{eq:12} gives
\begin{equation}\label{app_eq:10}
\frac{\partial\mathrm{FPL}_{2\mathrm{ph}}(1)}{\partial\theta^{(i)}_{l,m}}
=-\frac{8(\mathcal F^{(1),2\mathrm{ph}}_{\mathrm{e.d.}}-1)}{|\mathcal K|^2}
\sum_j\operatorname{Im}\left[e^{i\theta^{(i)}_{l,m}}
\langle m|B_l^{(i)}M_{ij}^{(2)}A_l^{(i)}|m\rangle\right],
\end{equation}
where
  \begin{equation}\label{app_eq:11}
    M^{(2)}_{ij} = \frac{1}{2} \left[(\mathrm{Tr} G_{ji})^2 + \mathrm{Tr} G_{ji}^{2} \right] \left[ (\mathrm{Tr} G_{ij})\mathbb I + U_{j}^{\dagger} U_i \right] U_{j}^{\dagger}.
\end{equation}

\section{Comparison of unitary 1-designs}\label{app:dif}
The   qudit Weyl operators are generated by \mbox{
\cite{anand2026qudit}
}\hskip0pt
\begin{eqnarray}
    && Z_d|j\rangle=\omega^j|j\rangle,\quad\omega=e^{2\pi i/d},\nonumber\\
    && X_d|j\rangle=|j+1\bmod d\rangle,\quad j\in\mathbb Z_d.
\end{eqnarray}
  Choosing representatives modulo scalar phase gives the unitary $1$-design
 \begin{equation}
    \mathcal{P}_d = \{X_d^{i}Z_d^{j}:\ i,j \in \mathbb{Z}_d \}.
\end{equation}

\begin{figure*}[ht!]
    \centering
    \includegraphics[width=\textwidth]{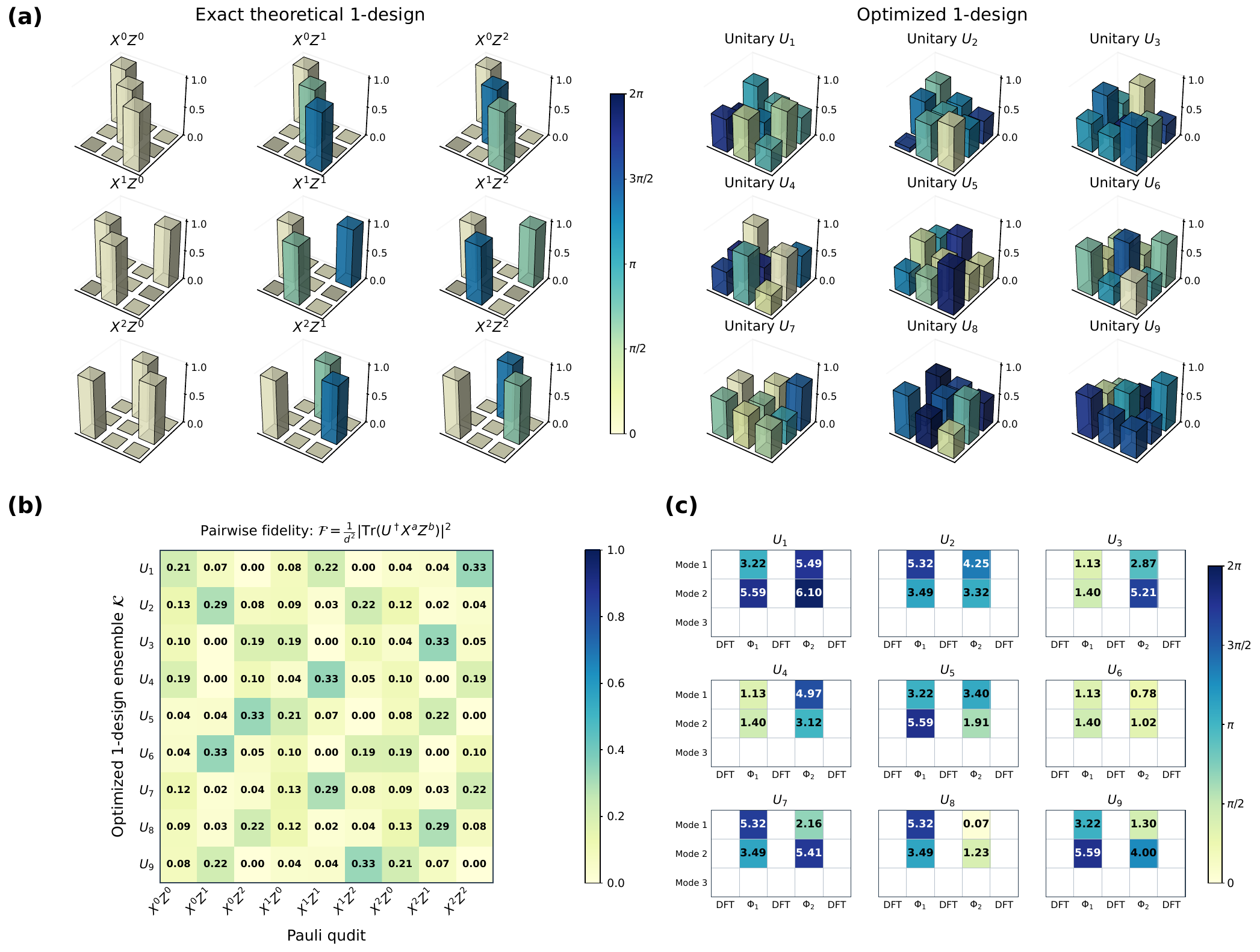}
    \caption{Exact and optimized unitary $1$-designs for $d=3$. (a) Weyl matrices $X^aZ^b$ (left) and optimized matrices $U_i$ with the smallest saved $\mathrm{FPL}(1)$ (right); bar heights and colors show moduli and phases. (b) Pairwise fidelities $|\operatorname{Tr}(U_i^\dagger X^aZ^b)|^2/9$. (c) Optimized phases in radians in the two-layer DFT circuit; the third-mode phase is fixed to zero in each layer.}
    \label{fig:3modes}
\end{figure*}

\FloatBarrier

\section{Representation theory}\label{app:repr}

  For a representation $A:G\to\mathrm{GL}(V)$, its character is the function \mbox{
\cite{fulton_representation}
}\hskip0pt
\begin{equation}\label{app_eq:12}
\chi_A(g)=\operatorname{Tr}A(g),\qquad g\in G.
\end{equation}

For   $n=1$, the character is $\chi_1(U)=\operatorname{Tr}U$. For $n=2$, the symmetric-power  character is

  \begin{equation}{\label{app_eq:13}}
    \chi_2(U) = \mathrm{Tr}\ R_2(U) = \frac{1}{2}\mathrm{Tr} \left(U \otimes U \left[\mathbb{I}+ \mathbb{F}\right] \right) = \frac{1}{2} \left[(\mathrm{Tr}U)^2 + \mathrm{Tr} U^2\right],
\end{equation}
where $\mathbb{I}$ is the identity permutation and $\mathbb{F}$ is the swap-operator, so $\mathrm{Tr} \left[(U \otimes U\ \mathbb{I})\right] = \mathrm{Tr} U \cdot \mathrm{Tr} U = (\mathrm{Tr} U)^2$, while $\mathrm{Tr} \left[(U \otimes U) \mathbb{F} \right] = \mathrm{Tr} U^2$, which follows from the swap-trick. Substituting this expression into the frame-potential formula, we obtain

\begin{equation}{\label{app_eq:14}}
    \mathcal{F}^{(t), 2ph}_{\mu_H} = \int_{U(d)} |\mathrm{Tr}\ \mathcal{U}^{(2)}|^{2t}\ d\mu_H(U) =  \int_{U(d)} |\chi_2(U)|^{2t}\ d\mu_H(U) = \frac{1}{4^{t}} \int_{U(d)} \left|(\mathrm{Tr}U)^2 + \mathrm{Tr}U^2 \right|^{2t}\ d\mu_H(U),
\end{equation}
  This is Eq.~ \eqref{eq:10}. For   three photons, the ambient projected operator is

  \begin{equation}{\label{app_eq:15}}
\widetilde{\mathcal U}^{(3)}
=\frac{1}{6}(U\otimes U\otimes U)\left[R_e+R_{(12)}+R_{(13)}+R_{(23)}+R_{(123)}+R_{(132)}\right].
\end{equation}
Here   $R_\pi$ permutes the three tensor factors. A transposition is a cycle of length two. The three transpositions are different operators, as are the two three-cycles. After taking the trace against $U^{\otimes3}$, permutations in the same conjugacy class contribute equally, giving
\begin{equation}
\chi_3(U)=\frac{1}{6}\left[(\operatorname{Tr}U)^3
+3\operatorname{Tr}U\operatorname{Tr}(U^2)
+2\operatorname{Tr}(U^3)\right].
\end{equation}

A useful tool in this context is the Newton--Girard identity, which allows one to compute the character for $n$ photons recursively:

\begin{equation}{\label{app_eq:16}}
    \chi_n(U)=\frac{1}n\sum_{k=1}^n\chi_{n-k}(U)\operatorname{Tr}(U^k),
    \qquad\chi_0(U)=1.
\end{equation}
This recursive relation makes it possible to evaluate frame potentials for arbitrary photon numbers and may be useful in further studies.

\section{Weingarten calculus}\label{app:wein}

The main tool for evaluating Haar integrals over matrix elements is the Weingarten formula derived in \cite{collins_sniady}:

\begin{equation}{\label{app_eq:17}}
        \int_{U(d)} U_{i_1, j_1} U_{i_2, j_2}\dots U_{i_{q}, j_{q}} U^{\dagger}_{l_1, m_1} U^{\dagger}_{l_2, m_2} \dots U^{\dagger}_{l_q, m_q} d\mu_H(U) = \sum_{\sigma, \tau \in S(q)} \delta_\sigma(\boldsymbol{i}, \boldsymbol{m}) \delta_{\tau} (\boldsymbol{j}, \boldsymbol{l}) \mathrm{Wg}(\sigma \tau^{-1}, d)
\end{equation}
  Here $\delta_\sigma(\boldsymbol i,\boldsymbol m)=\prod_{s=1}^q\delta_{i_s,m_{\sigma(s)}}$ and $\mathrm{Wg}$ is the unitary  Weingarten function \cite{weingarten_original}.
{
The character expansion of the Weingarten function can be evaluated using Young tableaux to describe the irreducible representations, the Murnaghan--Nakayama rule to compute symmetric-group characters, and Schur polynomials to obtain the corresponding unitary-group dimensions. For applications of Haar-moment methods to efficient unitary-design constructions, see Ref.~\cite{chen2024efficient}. Here we use only the second-moment case $q=2$, for $d\ge2$:
}

\begin{equation}{\label{app_eq:18}}
    \int_{U(d)} U_{i_1, j_1} U_{i_2, j_2} U^{\dagger}_{l_1, m_1} U^{\dagger}_{l_2, m_2} d\mu_H(U) = \mathrm{Wg}\left((1)(2), d\right) \left( \wg{2}{1, 2}{1, 2} + \wg{2}{2, 1}{2, 1} \right) + \mathrm{Wg}\left((12), d\right) \left( \wg{2}{2, 1}{1, 2} +  \wg{2}{1, 2}{2, 1} \right)
\end{equation}
where each graph represents the corresponding permutation contraction. For example, the graphical representations shown in Fig.~\ref{fig:wein_lol} correspond to the terms $\delta_{i_1, m_1} \delta_{i_2, m_2} \delta_{j_1, l_1}\delta_{j_2, l_2}$ and $\delta_{i_1, m_1} \delta_{i_2, m_2} \delta_{j_1, l_2}\delta_{j_2, l_1}$, respectively. Also, $\mathrm{Wg}\left((1)(2), d\right) = \frac{1}{d^2 - 1}$ and $\mathrm{Wg}\left((12), d\right) = -\frac{1}{d(d^2 - 1)}$, so

\begin{equation}{\label{app_eq:19}}
    \int_{U(d)} U_{i_1, j_1} U_{i_2, j_2} U^{\dagger}_{l_1, m_1} U^{\dagger}_{l_2, m_2} d\mu_H(U) = \frac{1}{d^2 - 1} \left( \wg{2}{1, 2}{1, 2} + \wg{2}{2, 1}{2, 1} \right) - \frac{1}{d(d^2 - 1)} \left( \wg{2}{2, 1}{1, 2} +  \wg{2}{1, 2}{2, 1} \right)
\end{equation}

\begin{figure}[h!]
  \centering
  \begin{minipage}{0.4\textwidth}
    \centering
    \drawWeingarten[$(1)(2)\rightarrow (1)(2): \pi = (1)(2)$]{2}{1, 2}{1, 2}
  \end{minipage}
  \hfill
  \begin{minipage}{0.4\textwidth}
    \centering
    \drawWeingarten[$(1)(2) \rightarrow (1 2): \pi = (12)$]{2}{1, 2}{2, 1}
  \end{minipage}
  \caption{Graphical representation of the two permutation-contraction patterns for two matrices $U$ and two matrices $U^{\dagger}$.}
  \label{fig:wein_lol}
\end{figure}

To obtain the value of $\mathcal{F}^{(1), 2\mathrm{ph}}_{\mu_H}$, we substitute $p_1 = \mathrm{Tr}\, U$ and $p_2 = \mathrm{Tr}\, U^2$ into \eqref{eq:10} for $t = 1$, which gives

\begin{equation}{\label{app_eq:20}}
    \mathcal{F}^{(1), 2\mathrm{ph}}_{\mu_H} = \frac{1}{4} \int_{U(d)} \left|p_1^2 + p_2 \right|^{2}\ d\mu_H(U) = \frac{1}{4} \int_{U(d)} \left(\underbrace{p_1^2 \overline{p_1}^2}_{I_1} + \underbrace{p_1^2 \overline{p_2}}_{I_2} + \underbrace{p_2 \overline{p_1}^2}_{I_3} + \underbrace{p_2 \overline{p_2}}_{I_4}\right) d\mu_H(U),
\end{equation}
where
\begin{eqnarray}{\label{app_eq:21}}
    &&I_1 = \int_{U(d)} p_1^2 \overline{p_1}^2\ d\mu_H(U) = \int_{U(d)} \left|\mathrm{Tr} U\right|^4 d\mu_H(U) = \mathcal{F}^{(2)}_{\mu_H} = 2 ,\nonumber\\
    &&I_2 = \int_{U(d)} p_1^2 \overline{p_2}\ d\mu_H = \sum_{i_1, i_2, j_1, j_2} \int_{U(d)}\ U_{i_1, i_1} U_{i_2, i_2} U^{\dagger}_{j_1, j_2} U^{\dagger}_{j_2, j_1}\ d\mu_H(U) = \nonumber\\
    && \quad \quad = \frac{2 \sum_{i_1, i_2, j_1, j_2} \delta_{i_1, i_2, j_1, j_2}}{d^2 - 1} - \frac{2 \sum_{i_1,j_1} \delta_{i_1, j_1} \sum_{i_2, j_2} \delta_{i_2, j_2}}{d(d^2 - 1)} = \frac{2d}{d^2 - 1} - \frac{2d^2}{d(d^2 - 1)}=0,\nonumber\\
    &&I_3 = \overline{I_2} = 0,\nonumber\\
    &&I_4 = \int_{U(d)} p_2 \overline{p_2}\ d\mu_H = \sum_{i_1, i_2, j_1, j_2} \int_{U(d)}\ U_{i_1, i_2} U_{i_2, i_1} U^{\dagger}_{j_1, j_2} U^{\dagger}_{j_2, j_1}\ d\mu_H(U) = \nonumber\\
    &&\quad \quad = \frac{2 \sum_{i_1,j_1} \delta_{i_1, j_1} \sum_{i_2, j_2} \delta_{i_2, j_2}}{d^2 - 1} - \frac{2 \sum_{i_1, i_2, j_1, j_2} \delta_{i_1, i_2, j_1, j_2}}{d(d^2 -1)} = \frac{2d^2}{d^2 - 1} - \frac{2d}{d(d^2 - 1)} = 2.
\end{eqnarray}
Therefore, we obtain $\mathcal{F}^{(1), 2\mathrm{ph}}_{\mu_H} = 1$.

\section{Twirling of quantum channels}\label{app:twirl}
The average gate fidelity of the channel from Eq. (\ref{eq:13}) with respect to $V$ is defined as the integral over all pure input states $|\psi\rangle = U|0\rangle$, uniformly distributed according to the Haar measure $d\mu_H(U)$:
\begin{equation}
    \langle F_{\mathcal{E}} \rangle = \int_{U(d)} \langle 0 | U^\dagger V^\dagger \mathcal{E}\left(U |0\rangle \langle 0| U^\dagger\right) V U |0\rangle d\mu_H(U).
\end{equation}
Substituting the channel definition and applying Fubini's theorem to swap the order of integration, we obtain
\begin{equation}
    \langle F_{\mathcal{E}} \rangle = \int d\boldsymbol{\phi} P(\boldsymbol{\phi}) \underbrace{\int_{U(d)} \left| \langle 0 | U^\dagger \Phi_{\boldsymbol{\phi}} U |0\rangle \right|^2 d\mu_H(U)}_{I_H(\boldsymbol{\phi})},
\end{equation}
where we have introduced the auxiliary unitary operator $\Phi_{\boldsymbol{\phi}} = V^\dagger N(\boldsymbol{\phi}) V$. Rewriting out the expression under the integral $I_H(\boldsymbol{\phi})$ component by component and letting the matrix elements of $\Phi_{\boldsymbol{\phi}}$ in the standard basis be $\Phi_{i,j}$:

\begin{equation}
    \left| \langle 0 | U^\dagger \Phi_{\boldsymbol{\phi}} U |0\rangle \right|^2 = \sum_{i_1, j_1, i_2, j_2} \Phi_{i_1, j_1} \Phi_{i_2, j_2}^* U_{i_1, 0}^* U_{j_1, 0} U_{i_2, 0} U_{j_2, 0}^*,
\end{equation}

so, using (\ref{app_eq:19}), we obtain

{
The four matrix elements in the integral have two unconjugated factors, $U_{j_1,0}$ and $U_{i_2,0}$, and two conjugated factors, $\overline{U_{i_1,0}}$ and $\overline{U_{j_2,0}}$. Equation~\eqref{app_eq:19} therefore gives
\begin{multline}
\int_{U(d)} U_{j_1,0}U_{i_2,0}
 \overline{U_{i_1,0}}\overline{U_{j_2,0}}\,d\mu_H(U)
=\left(\frac{1}{d^2-1}-\frac{1}{d(d^2-1)}\right)
 \left(\delta_{j_1,i_1}\delta_{i_2,j_2}
       +\delta_{j_1,j_2}\delta_{i_2,i_1}\right)=\\
\frac{\delta_{j_1,i_1}\delta_{i_2,j_2}
       +\delta_{j_1,j_2}\delta_{i_2,i_1}}{d(d+1)}.
\end{multline}
Substituting this integral and carrying out the Kronecker-delta contractions step by step, we find
\begin{multline}
I_H(\boldsymbol\phi)
=\sum_{i_1,j_1,i_2,j_2}\Phi_{i_1,j_1}\overline{\Phi_{i_2,j_2}}
 \int_{U(d)}U_{j_1,0}U_{i_2,0}
 \overline{U_{i_1,0}}\overline{U_{j_2,0}}\,d\mu_H(U)=\\
\frac{1}{d(d+1)}\sum_{i_1,j_1,i_2,j_2}
 \Phi_{i_1,j_1}\overline{\Phi_{i_2,j_2}}
 \left(\delta_{j_1,i_1}\delta_{i_2,j_2}
       +\delta_{j_1,j_2}\delta_{i_2,i_1}\right)
=\frac{1}{d(d+1)}\left[
 \sum_{i_1,i_2}\Phi_{i_1,i_1}\overline{\Phi_{i_2,i_2}}
 +\sum_{i_1,j_1}\Phi_{i_1,j_1}\overline{\Phi_{i_1,j_1}}\right]=\\
\frac{|\operatorname{Tr}\Phi_{\boldsymbol\phi}|^2
       +\sum_{i_1,j_1}|\Phi_{i_1,j_1}|^2}{d(d+1)}
=\frac{|\operatorname{Tr}\Phi_{\boldsymbol\phi}|^2
       +\operatorname{Tr}(\Phi_{\boldsymbol\phi}\Phi_{\boldsymbol\phi}^{\dagger})}{d(d+1)}
=\frac{|\operatorname{Tr}\Phi_{\boldsymbol\phi}|^2+d}{d(d+1)}.
\end{multline}

}

By the cyclic property of the trace, $\Tr(\Phi_{\boldsymbol{\phi}}) = \Tr(V^\dagger N(\boldsymbol{\phi}) V) = \Tr(N(\boldsymbol{\phi})) = \sum_{k=1}^d e^{i\phi_k}$. Therefore, the exact fidelity under a fixed realization of the phase noise is
\begin{equation}
    I_H(\boldsymbol{\phi}) = \frac{1}{d(d+1)} \left[ \left| \sum_{k=1}^d e^{i\phi_k} \right|^2 + d \right].
\end{equation}

Finally, we perform the statistical averaging over the Gaussian noise distribution:
\begin{equation}
    \mathbb{E}_{\boldsymbol{\phi}} \left[ \left| \sum_{k=1}^d e^{i\phi_k} \right|^2 \right] = \sum_{k=1}^d \sum_{m=1}^d \mathbb{E}\left[ e^{i(\phi_k - \phi_m)} \right] = d + \sum_{k \neq m} \mathbb{E}\left[ e^{i(\phi_k - \phi_m)} \right].
\end{equation}
For the off-diagonal terms ($k \neq m$), the difference $\Delta = \phi_k - \phi_m$ is a normally distributed random variable with zero mean and variance $2\sigma^2$. The expectation value of the exponential corresponds to the characteristic function of the normal distribution evaluated at $t=1$:
\begin{equation}
    \mathbb{E}\left[ e^{i\Delta} \right] = e^{-\frac{1}{2}\text{Var}(\Delta)} = e^{-\sigma^2}.
\end{equation}
Given that there are $d(d-1)$ off-diagonal pairs, the expectation value simplifies to
\begin{equation}
    \mathbb{E}_{\boldsymbol{\phi}} \left[ \left| \sum_{k=1}^d e^{i\phi_k} \right|^2 \right] = d + d(d-1)e^{-\sigma^2}.
\end{equation}
Substituting this statistical average back into the expression for $\langle F_{\mathcal{E}} \rangle$, we obtain
\begin{equation}
    \langle F_{\mathcal{E}} \rangle = \frac{d + d(d-1)e^{-\sigma^2} + d}{d(d+1)} = \frac{d \left( 2 + (d-1)e^{-\sigma^2} \right)}{d(d+1)}.
\end{equation}
Factoring out the dimension $d$, we arrive at the final analytical formula for the average gate fidelity:
\begin{equation}
    \langle F_{\mathcal{E}} \rangle = \frac{2 + (d-1)e^{-\sigma^2}}{d+1}.
\end{equation}

\section{Photon-number estimation}\label{app:estim}

The derivation of Eq.~\eqref{eq:16} is based on the transformation of creation and annihilation operators under a $d$-mode unitary matrix $U$
\begin{equation}{\label{app_eq:25}}
    \hat{a}^{\dagger}_{k} \rightarrow \sum_{i = 1}^{d} U_{ki}\ \hat{a}^{\dagger}_i, \quad \hat{a}_{k} \rightarrow \sum_{j = 1}^{d} U^{\dagger}_{jk}\ \hat{a}_{j}.
\end{equation}
So, substituting these transformations into \eqref{eq:16}, assuming the input state as $\ket{\psi} = \ket{n_1, n_2, \dots, n_d}$ with $\sum_{i} n_i = n$, we will obtain
\begin{multline}{\label{app_eq:26}}
    \mathbb{E}_{\mu_H} \left[ \langle \hat{n}_k \rangle \right] = \mathbb{E}_{\mu_H} \left[\sum_{i,j = 1}^{d} U_{ki} U^{\dagger}_{jk} \langle \hat{a}^{\dagger}_{i}\hat{a}_{j} \rangle \right] = \mathbb{E}_{\mu_H} \left[\sum_{i,j = 1}^{d} U_{ki} U_{jk}^{\dagger} n_{i} \delta_{ij} \right] = \sum_{i = 1}^{d} n_i \underbrace{\int_{U(d)} U_{ki} U_{ik}^{\dagger} d\mu_H(U)}_{1-\mathrm{design}}=\\
    = \sum_{i = 1}^{d} n_i \frac{1}{d} \delta_{kk} \delta_{ii} = \frac{n}{d}.
\end{multline}

For the following equation \eqref{eq:17}, we can rewrite
\begin{equation}{\label{app_eq:27}}
    \hat{n}_k^{2} = \hat{a}^{\dagger}_{k}\hat{a}_{k} \hat{a}^{\dagger}_k \hat{a}_k = (\hat{a}^{\dagger}_k)^2 \hat{a}_k^2 + \hat{n}_k,
\end{equation}
so we first evaluate the normally ordered contribution to the second moment at the output of the interferometer
\begin{multline}{\label{app_eq:28}}
    \mathbb{E}_{\mu_H} \left[\sum_{i, j, m,l = 1}^{d} U_{ki} U_{kj} U^{\dagger}_{mk} U^{\dagger}_{lk} \langle \hat{a}^{\dagger}_i \hat{a}^{\dagger}_j \hat{a}_m \hat{a}_l\rangle \right] = \\
    = \mathbb{E}_{\mu_H} \left[\sum_{i,j,m,l = 1}^{d} U_{ki} U_{kj} U^{\dagger}_{mk} U^{\dagger}_{lk} \left( n_i n_j\underbrace{(\delta_{il} \delta_{jm} + \delta_{im} \delta_{jl})}_{i \neq j} + (n_i^2 - n_i) \delta_{ijml}\right)\right] =  \\
    = 2 \mathbb{E}_{\mu_H} \left[\sum_{i\neq j}^{d} n_i n_j U_{ki} U_{kj} U^{\dagger}_{jk} U^{\dagger}_{ik}\right] + \mathbb{E}_{\mu_H} \left[\sum_{i = 1}^{d} (n_i^2 - n_i) U_{ki} U_{ki} U^{\dagger}_{ik} U^{\dagger}_{ik} \right] =  2 \sum_{i\neq j}^{d} n_i n_j \underbrace{\int_{U(d)} U_{ki} U_{kj} U^{\dagger}_{jk} U^{\dagger}_{ik}\ d\mu_H(U)}_{2-\mathrm{design}}  +  \\
    + \sum_{i = 1}^{d} (n_i^2 - n_i) \underbrace{\int_{U(d)}U_{ki} U_{ki} U^{\dagger}_{ik} U^{\dagger}_{ik}\ d\mu_H(U)}_{2-\mathrm{design}} = 2\sum_{i\neq j}^{d} n_i n_j\delta_{kk} \delta_{ii} \delta_{kk} \delta_{jj} \left[\frac{1}{d^2 - 1}  - \frac{1}{d(d^2 - 1)} \right] + \\
    + 2\sum_{i = 1}^{d} (n_i^2 - n_i) \delta_{kk} \delta_{ii} \left[\frac{1}{d^2 - 1}- \frac{1}{d(d^2 - 1)}\right] = \\
    = \frac{2}{d(d+1)} \sum_{i = 1}^{d} n_i (n - n_i) + \frac{2}{d(d+1)} \left[ \sum_{i = 1}^{d} n_i^2 - \sum_{i = 1}^{d} n_i \right] = \frac{2}{d(d+1)} \left(n^2 - n \right)
\end{multline}

Consequently, the Haar-averaged second moment is
\begin{equation}\label{app_eq:number_second_moment}
    \mathbb{E}_{\mu_H}\left[\langle\hat n_k^2\rangle_U\right]
    =\frac{2n(n-1)}{d(d+1)}+\frac nd.
\end{equation}

The variance in Eq.~\eqref{eq:17} refers to photon counts pooled over the sampled interferometers and detection outcomes. It contains both the quantum variance at fixed $U$ and the fluctuations of the quantum mean between different $U$. To evaluate these contributions explicitly, we first square the quantum mean at fixed $U$ and then perform the Haar average. Using $\langle\hat n_k\rangle_U=\sum_{i=1}^d n_i U_{ki}U^\dagger_{ik}$, we obtain
\begin{multline}\label{app_eq:number_mean_squared}
    \mathbb{E}_{\mu_H}\left[\langle\hat n_k\rangle_U^2\right]
    =\mathbb{E}_{\mu_H}\left[\left(\sum_{i=1}^d n_i U_{ki}U^\dagger_{ik}\right)^2\right] = \\
    =\mathbb{E}_{\mu_H}\left[\sum_{i,j=1}^d n_i n_j U_{ki}U^\dagger_{ik}U_{kj}U^\dagger_{jk}\right] = \\
    =\sum_{i,j=1}^d n_i n_j
    \underbrace{\int_{U(d)}U_{ki}U_{kj}U^\dagger_{ik}U^\dagger_{jk}\,d\mu_H(U)}_{2-\mathrm{design}} = \\
    =\sum_{i,j=1}^d n_i n_j\frac{1+\delta_{ij}}{d(d+1)}
    =\frac{n^2+\sum_{i=1}^d n_i^2}{d(d+1)}.
\end{multline}

Thus, the Haar average of the quantum variance $\mathrm{Var}_U(\hat n_k)=\langle\hat n_k^2\rangle_U-\langle\hat n_k\rangle_U^2$ is
\begin{multline}\label{app_eq:conditional_variance}
    \mathbb{E}_{\mu_H}\left[\mathrm{Var}_U(\hat n_k)\right]
    =\frac{2n(n-1)}{d(d+1)}+\frac nd
    -\frac{n^2+\sum_{i=1}^d n_i^2}{d(d+1)} = \\
    =\frac{n(d-1)+n^2-\sum_{i=1}^d n_i^2}{d(d+1)}.
\end{multline}
The fluctuations of the quantum mean around the Haar mean $n/d$ contribute
\begin{multline}\label{app_eq:mean_fluctuations}
    \mathbb{E}_{\mu_H}\left[\left(\langle\hat n_k\rangle_U-\frac nd\right)^2\right]
    =\mathbb{E}_{\mu_H}\left[\langle\hat n_k\rangle_U^2\right]-\frac{n^2}{d^2} = \\
    =\frac{n^2+\sum_{i=1}^d n_i^2}{d(d+1)}-\frac{n^2}{d^2}
    =\frac{d\sum_{i=1}^d n_i^2-n^2}{d^2(d+1)}.
\end{multline}

Adding these two contributions gives the law of total variance. Equivalently, we can keep the quantum expectation inside the Haar average and center the count operator at the ensemble mean $n/d$:
\begin{multline}\label{app_eq:29}
    \mathrm{Var}_{\mu_H}(\hat n_k)
    =\mathbb{E}_{\mu_H}\left[\mathrm{Var}_U(\hat n_k)\right]
    +\mathbb{E}_{\mu_H}\left[\left(\langle\hat n_k\rangle_U-\frac nd\right)^2\right] = \\
    =\mathbb{E}_{\mu_H}\left[\left\langle\left(\hat n_k-\frac nd\right)^2\right\rangle_U\right]
    =\mathbb{E}_{\mu_H}\left[\langle\hat n_k^2\rangle_U\right]
    -\left(\mathbb{E}_{\mu_H}\left[\langle\hat n_k\rangle_U\right]\right)^2 = \\
    =\frac{2n(n-1)}{d(d+1)}+\frac nd-\frac{n^2}{d^2}
    =\frac{n(d-1)(n+d)}{d^2(d+1)}.
\end{multline}
The dependence on the individual input occupations $n_i$ cancels between the two contributions, so the final result depends only on the total photon number $n$ and the number of modes $d$.

\end{document}